\documentclass[]{elsarticle}
\usepackage{geometry}                		% See geometry.pdf to learn the layout options. There are lots.
\usepackage{graphicx}				% Use pdf, png, jpg, or eps§ with pdflatex; use eps in DVI mode
\usepackage{amssymb}
\usepackage{amsmath}
\usepackage{amsthm}
\usepackage{bm}
\usepackage{stackrel}
\usepackage{mathrsfs}
\usepackage[normalsize]{subfigure}
\usepackage{mathtools}
\newcommand{\Sl}[1]{{F}^{#1}}

\usepackage{thmtools}
\usepackage{thm-restate}
\usepackage{hyperref}
\usepackage{enumitem}

\journal{Performance Evaluation}

\begin{document}
\begin{frontmatter}

\title{heSRPT:Parallel Scheduling to Minimize Mean Slowdown}

\author[cmu]{Benjamin Berg \corref{bsb}}
\cortext[bsb]{Corresponding Author.  This author was supported by a Facebook Graduate Fellowship.}
\ead{bsberg@cs.cmu.edu}

\author[mu]{Rein Vesilo}
\ead{rein.vesilo@mq.edu.au}

\author[cmu]{Mor Harchol-Balter \corref{mhb}}
\cortext[mhb]{This author was supported by: NSF-CMMI-1938909, NSF-CSR-1763701, NSF-XPS-1629444, and a Google 2020 Faculty Research Award.}
\ead{harchol@cs.cmu.edu}

\address[cmu]{Carnegie Mellon University}
\address[mu]{Macquarie University}

							% Activate to display a given date or no date
\begin{abstract}
Modern data centers serve workloads which are capable of exploiting parallelism.  
When a job parallelizes across multiple servers it will complete more quickly, but jobs receive diminishing returns from being allocated additional servers.
Because allocating multiple servers to a single job is inefficient, it is unclear how best to allocate a fixed number of servers between many parallelizable jobs.

This paper provides the first optimal allocation policy for minimizing the mean slowdown of parallelizable jobs of known size when all jobs are present at time 0.
Our policy provides a simple closed form formula for the optimal allocations at every moment in time.
Minimizing mean slowdown usually requires favoring {\em short} jobs over long ones (as in the SRPT policy).   
However, because parallelizable jobs have sublinear speedup functions, system efficiency is also an issue.   
System efficiency is maximized by giving \emph{equal} allocations to all jobs and thus competes with the goal of prioritizing small jobs.  
Our optimal policy, high-efficiency SRPT (heSRPT), balances these competing goals.
heSRPT completes jobs according to their size order, but maintains overall system efficiency by allocating some servers to each job at every moment in time.
Our results generalize to also provide the optimal allocation policy with respect to mean flow time.

Finally, we consider the online case where jobs arrive to the system over time.
While optimizing mean slowdown in the online setting is even more difficult, we find that heSRPT provides an excellent heuristic policy for the online setting.
In fact, our simulations show that heSRPT significantly outperforms state-of-the-art allocation policies for parallelizable jobs.
\end{abstract}


\end{frontmatter}

\section{Introduction}
Modern data centers serve workloads which are capable of exploiting parallelism.  
When a job parallelizes across multiple servers it will complete more quickly.
However, it is unclear how to share a limited number of servers between many parallelizable jobs.

In this paper we consider a typical scenario where a data center composed of $N$ servers will be tasked with completing a set of $M$ parallelizable jobs, where typically $M$ is much smaller than $N$.
In our scenario, each job has a different inherent size (service requirement) which is known up front to the system.
In addition, each job can utilize any number of servers at any moment in time.
These assumptions are reasonable for many parallelizable workloads such as training neural networks using TensorFlow \cite{abadi2016tensorflow,lin2018model}. 
Our goal in this paper is to allocate servers to jobs so as to minimize the \emph{mean slowdown} across all jobs, where the slowdown of a job is the job's completion time divided by its running time if given exclusive access to all $N$ servers.
Slowdown is a measure of how a job was interfered with by other jobs in the system, and is often the metric of interest in the theoretical parallel scheduling literature (where it is also called stretch) \cite{muthukrishnan1999online}, as well as the HPC community (where it is called expansion factor) \cite{jackson2000simulation}.

\begin{figure}[t]
\centering
\includegraphics[width=.70\textwidth]{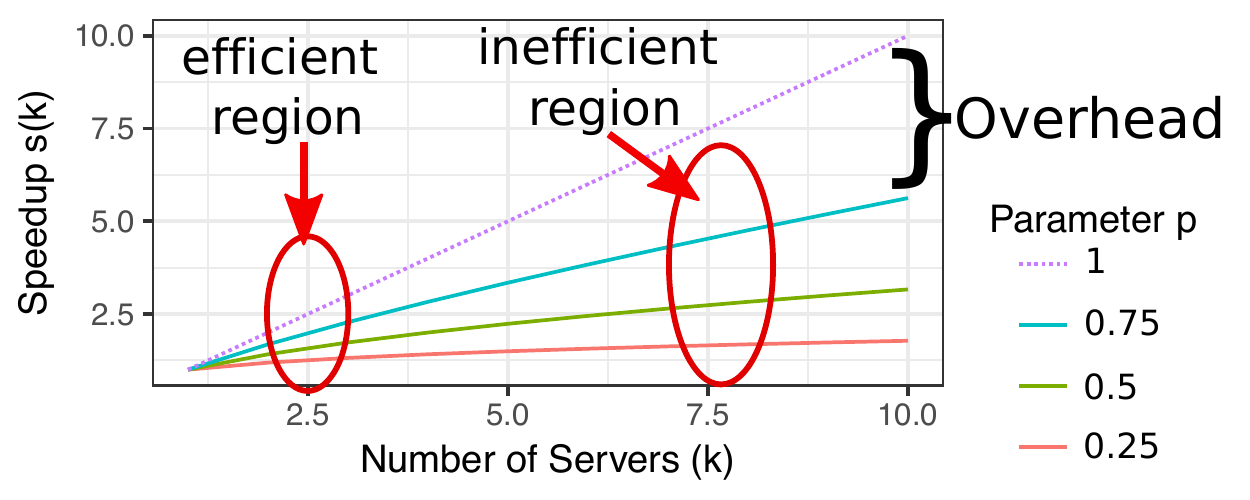}
    \caption{A variety of speedup functions of the form $s(k)=k^p$, shown with varying values of $p$.
    When $p=1$ we say that jobs are \emph{embarrassingly parallel}, and hence we consider cases where $0 < p <1$.
    Note that all functions in this family are concave and lie below the embarrassingly parallel speedup function ($p=1$).}

\label{fig:speedups}
\end{figure}

What makes this problem difficult is that jobs receive a concave, sublinear speedup from parallelization -- jobs have a decreasing marginal benefit from being allocated additional servers (see Figure \ref{fig:speedups}).
Hence, in choosing a job to receive each additional server, one must keep the overall efficiency of the system in mind.
The goal of this paper is to determine the optimal allocation of servers to jobs where all jobs follow a realistic sublinear speedup function.

It is clear that the optimal allocation policy will depend heavily on the jobs' speedup -- how parallelizable the jobs being run are.  
To see this, first consider the case where each job is \emph{embarrassingly parallel} (see Figure \ref{fig:speedups}), and can be parallelized perfectly across an arbitrary number of servers.  
In this case, we observe that the entire data center can be viewed as a \emph{single server} that can be perfectly utilized by or shared between jobs.
Hence, from the single server scheduling literature, it is known that the Shortest Remaining Processing Time policy (SRPT) will minimize the mean slowdown across jobs \cite{smith1978new}.
By contrast, if we consider the case where jobs are hardly parallelizable, a single job receives very little benefit from additional servers.
In this case, the optimal policy is to divide the system equally between jobs, a policy called EQUI.  
In practice, a realistic speedup function usually lies somewhere between these two extremes and thus we must balance a trade-off between the SRPT and EQUI policies in order to minimize mean slowdown.  
Specifically, since jobs are \emph{partially} parallelizable, it is still beneficial to allocate more servers to smaller jobs than to large jobs.  
The optimal policy with respect to mean slowdown must split the difference between these policies, figuring out how to favor short jobs while still respecting the overall efficiency of the system.

In this paper, we present the optimal allocation policy with respect to mean slowdown, which balances the tradeoff between EQUI and SRPT.
We call this policy \emph{high efficiency SRPT}.
Our analysis considers the case where all jobs are present in the system at time 0.
While this is a common case in practice, it is also interesting to consider an online version of the problem, where jobs arrive over time.
Unfortunately, it has been shown that, in general, no optimal policy exists to minimize mean slowdown in the online case \cite{bansal2003server}.
We demonstrate that heSRPT, which is optimal in the offline case, provides an excellent heuristic policy, Adaptive-heSRPT, for the online case. Adaptive-heSRPT often performs an order of magnitude better than policies previously suggested in the literature.

\begin{figure}[t]
\centering
\includegraphics[width=.75\textwidth]{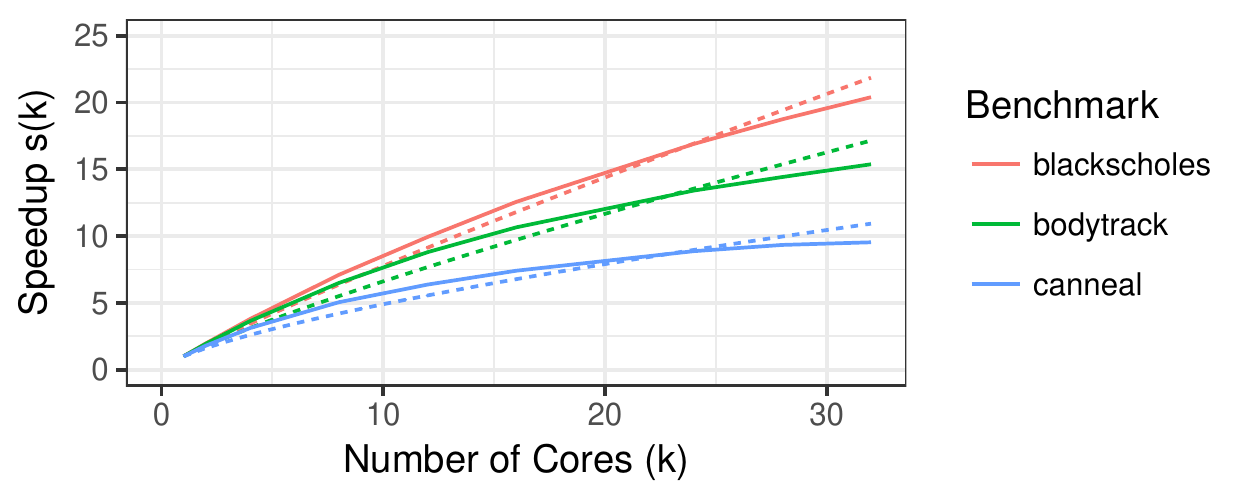}
    \caption{Various speedup functions of the form $s(k)=k^p$ (dotted lines) which have been fit to real speedup curves (solid lines) measured from jobs in the PARSEC-3 parallel benchmarks\cite{zhan2017parsec3}.  The three jobs, blackscholes, bodytrack, and canneal, are best fit by the functions where $p=.89$, $p=.82$, and $p=.69$ respectively.}

\label{fig:parsec}
\end{figure}

\subsection*{Our Model}
Our model assumes there are $N$ identical servers which must be allocated to $M$ parallelizable jobs.  All $M$ jobs are present at time $t=0$.
Job $i$ is assumed to have some inherent size $x_i$ where, without loss of generality (WLOG),
$$x_1 \geq x_2 \geq \ldots \geq x_M.$$

In general we will assume that all jobs follow the \emph{same} speedup function, $s : \mathbb{R}^+ \rightarrow \mathbb{R}^+$, which is of the form
$$s(k) = k^p \qquad$$
for some $0 <p<1$.
Specifically, if a job $i$ of size $x_i$ is allocated $k$ servers, it will complete at time
$$\frac{x_i}{s(k)}.$$
In general, the number of servers allocated to a job can change over the course of the job's lifetime.
It therefore helps to think of $s(k)$ as a \emph{rate}\footnote{WLOG we assume the service rate of a single server to be 1.  More generally, we could assume the rate of each server to be $\mu$, which would simply replace $s(k)$ by $s(k)\mu$ in every formula.} of service where the remaining size of job $i$ after running on $k$ servers for a length of time $t$ is
$$x_i - t \cdot s(k).$$
We choose the family of functions $s(k)=k^p$ because they are (i) sublinear and concave, (ii) can be fit to a variety of empirically measured speedup functions (see Figure \ref{fig:parsec}) \cite{zhan2017parsec3}, and (iii) simplify the analysis.  Note that \cite{Hill:2008:ALM:1449375.1449387} assumes $s(k) = k^p$ where $p=0.5$ and explicitly notes that using speedup functions of another form does not significantly impact their results. %add Hill et. al to prior work.

In general, we assume that there is some policy, $P$, which allocates servers to jobs at every time, $t$.
When describing the state of the system, we will use $m^P(t)$ to denote the number of remaining jobs in the system at time $t$, and $x^P_i(t)$ to denote the remaining size of job $i$ at time $t$.
To describe the state of each job at time $t$, let $W^P_i(t)$ denote the total amount of work done by policy $P$ on job $i$ by time $t$.
Let $W^P_i(t_1,t_2)$ be the amount of work done by policy $P$ on job $i$ on the interval $[t_1,t_2)$. 
That is, 
$$W^P_i(t_1,t_2) = W^P_i(t_2)-W^P_i(t_1).$$
We denote the completion time of job $i$ under policy $P$ as $T^P_i$.
%Hence, the slowdown of job $i$ under policy $P$, $S^P_i$, can be written as
%$$S^P_i = \frac{T^P_i}{x_i}.$$
When the policy $P$ is implied, we will drop the superscript.

We will assume that the number of servers allocated to a job need not be discrete.
In general, we will think of the $N$ servers as a \emph{single, continuously divisible resource}. 
Hence, the policy $P$ can be defined by an \emph{allocation function} $\bm{\theta}^P(t)$ where
$$\bm{\theta}^P(t) = (\theta^P_{1}(t), \theta^P_{2}(t), \ldots, \theta^P_{M}(t)).$$
Here, $0 \leq \theta^P_i(t) \leq 1$ for each job $i$, and $\sum_{i=1}^{M} \theta_i^P(t) \leq 1$. 
%\footnote{Any policy for which $\sum_{i=1}^{M} \theta_i^P(t) < 1$ can be easily improved, but arguing about such policies will be helpful throughout the paper}.
An allocation of $\theta^P_{i}(t)$ denotes that under policy $P$, at time $t$, job $i$ receives a speedup of $s(\theta^P_{i}(t)\cdot N)$.  Completed jobs are allocated 0 servers.

For each job, $i$, we define a corresponding \emph{weight} $w_i > 0$.
We then define the \emph{weighted flow time} for a set of jobs under policy $P$, $F^P$, to be
$$\Sl{P} = \sum_{i=1}^M w_i \cdot T_i.$$
By manipulating the weights associated with each job, one can derive metrics with different intuitive meanings.
We say that weights \emph{favor small jobs} if larger weights are always assigned to smaller jobs.
That is, if
$$w_1 \leq w_2 \leq \ldots \leq w_M.$$

The class of weighted flow time metrics where weights favor small jobs includes several popular metrics.
We are primarily interested in the mean slowdown of a policy $P$, $\overline{S}^P$, which is defined as
$$\overline{S}^P = \frac{1}{M}\cdot \sum_{i=1}^M \frac{T^P_i}{x_i/s(N)}.$$
The objective of minimizing mean slowdown is equivalent to minimizing weighted flow time when $w_i=\frac{1}{x_i/s(N)}$.
Clearly, weights favor small jobs in this setting.
Similarly, the objective of minimizing mean flow time is equivalent to minimizing weighted flow time when $w_i=1$ for each job, $i$.

For the sake of generality, we consider the problem of finding the policy $P^*$ which minimizes the weighted flow time for a set of $M$ jobs when weights favor small jobs.
Let $F^*$ be the weighted flow time under this optimal policy.
We will denote the allocation function of the optimal policy as $\bm{\theta}^*(t)$.  Similarly, we let $m^*(t)$, $x^*_i(t)$, $W^*_i(t)$ and $T^*_i$ denote the corresponding quantities under the optimal policy.

\subsection*{Why Server Allocation is Counter-intuitive}
Consider a simple system with $N=10$ servers and $M=2$ identical jobs of size 1, where $s(k) = k^{.5}$, and where we wish to minimize mean flow time (which is equivalent to slowdown in this case).
A queueing theorist might look at this problem and say that to minimize flow time, we should use the SRPT policy, first allocating all servers to job one and then all servers to job two.
However, this causes the system to be very inefficient.
Another intuitive argument would be that, since everything in this system is symmetric, the optimal allocation should be symmetric.
Hence, one might think to allocate half the servers to job one and half the servers to job two.
While this equal allocation does maximize the efficiency of the system by ensuring that neither job receives too many servers, it does \emph{not} minimize their mean flow time.
Theorem \ref{thm:opt} will show that the optimal policy in this case is to allocate $75\%$ of the servers to job one and $25\%$ of the servers to job two.
In our simple, symmetric system, the optimal allocation is very asymmetric!
Note that this asymmetry is \emph{not} an artifact of the form of the speedup function used.
For example, if we had instead assumed that the speedup function $s$ was Amdahl's Law \cite{Hill:2008:ALM:1449375.1449387} with a parallelizable fraction of $f=.9$, the optimal split is to allocate $63.5\%$ of the system to one of the jobs.
Instead, the optimal policy balances the tradeoff between using an efficient equal allocation and using the inefficient SRPT policy which favors small jobs.
If we imagine a set of $M$ arbitrarily sized jobs, one suspects that the optimal policy again favors shorter jobs, but it is not obvious how to calculate the exact allocations for this policy.

\subsection*{Why Finding the Optimal Policy is Hard}
At first glance, solving for the optimal policy seems amenable to classical optimization techniques.
However, naive application of these techniques would require solving $M!$ optimization problems, each consisting of $O(M^2)$ variables and $O(M)$ constraints.
Furthermore, although these techniques could produce the optimal policy for a single problem instance, it is unlikely that they would yield a closed form solution.
We instead advocate for finding a closed form solution for the optimal policy, which allows us to build intuition about the underlying dynamics of the system.

To understand the source of this complexity, we will first consider the problem of minimizing the total flow time of a set of just two jobs of sizes $x_1$ and $x_2$ respectively.
If we assume that the optimal policy completes job 2 first, we can write an expression for the total flow time of the two jobs as follows:
$$T_1 + T_2 = \frac{2 \cdot x_2}{s(\theta_2(0))} + \frac{\left(x_1 - \frac{s(\theta_1(0))\cdot x_2}{s(\theta_2(0))}\right)}{s(N)}.$$
The first term in this expression describes the total flow time accrued between time 0 and time $T_2$.
The duration of this period is $\frac{x_2}{s(\theta_2(0))}$, and because there are two jobs in the system during this period, the total flow time accrued during the period is $\frac{2 \cdot x_2}{s(\theta_2(0))}$.
The second term encodes the remaining time required to finish job 1 after job 2 completes.
This expression assumes that job allocations only change at the time of a departure, which we will prove formally in Theorem \ref{thm:constant}.
More interestingly, however, this expression \emph{implicitly encodes} the idea that the policy finishes job 2 before job 1.
If, on the other hand, job 1 finishes first, the expression would change to
$$T_1 + T_2 = \frac{2 \cdot x_1}{s(\theta_1(0))} + \frac{\left( x_2 - \frac{s(\theta_2(0))\cdot x_1}{s(\theta_1(0))}\right)}{s(N)}.$$

This has two main implications.
First, it is possible that the allocation policy which minimizes a particular expression for total flow time does not complete jobs in the order encoded in the expression.
In this case, the value of the expression is meaningless -- it does not equal the total flow time of the jobs under the computed policy.
Hence, to find the optimal policy with respect to a particular completion order, one must minimize the corresponding expression for total flow time subject to constraints which maintain the completion order of the allocation policy.
Second, because the objective function depends on the completion order of the jobs, there is no single function to try to minimize to recover an optimal policy.
Instead, we are interested in the \emph{global} minimum across all of the $M!$ completion orders, each of which has its own \emph{local} minimum that can be obtained by constrained minimization.

When there are only two jobs, it is tractable to directly solve for the optimal policy via known constrained optimization methods.
One can  use Lagrange multipliers with two constraints -- one to ensure a valid allocation policy and one to enforce the completion order.
One must solve two such optimization problems corresponding to each of the potential completion orders.
However, given a set of $M$ jobs with $M!$ possible completion orders, this naive approach would require solving $M!$ optimization problems, each consisting of $O(M^2)$ variables used to define an allocation policy and $O(M)$ constraints.
Even if some of this complexity could be elided by, for instance, proving the completion order of the optimal policy (see Theorem \ref{thm:sjf}), Lagrange multipliers are unlikely to emit a closed form given such a complex objective function and large set of constraints.

Additionally, one may think to apply classical techniques from the deterministic scheduling literature.
Specifically, one technique is to show that the problem of minimizing weighted flow time can be reduced to solving a linear program where the feasible region is a polymatroid \cite{yao2002dynamic,bertsimas1996conservation}.
The polymatroid technique can be used to show that a system is \emph{indexable} (the optimal policy is a simple priority policy) or even \emph{decomposable} (job priorities do not depend on the other jobs in the system).
This technique can be used to provide a greedy algorithm for obtaining the optimal policy.
Unfortunately, the polymatroid technique is not straightforward to apply in our case.

Specifically, one generally proceeds by showing that a problem satisfies the so-called \emph{generalized conservation laws} \cite{yao2002dynamic,bertsimas1996conservation}.
However, these conservation laws require scheduling policies which are \emph{work-conserving}, meaning the total rate at which the system completes work remains constant over time.
Allocation policies in our setting are not work conserving in general, and the allocations used by the optimal policy do not maintain a constant work rate (see Figure \ref{fig:3jobs}).
Additionally, due to the form of the speedup function $s(k)=k^p$, the region of achievable weighted flow times is not a polytope.
Furthermore, the optimal policy we derive in Theorem \ref{thm:opt} will show that our problem is not decomposable for the case of minimizing weighted flow time where weights favor small jobs.
While there may exist a reduction of our problem that would allow one to establish some form of conservation laws, it is far from obvious what the ``conserved quantity'' would be for our problem which would allow for application of these techniques.

Prior work has investigated the related problem of scheduling flows in networks where the system capacity evolves over time \cite{sadiq2010balancing}.
However, this work derives an optimal policy only in the case when the work rates of each job are constrained to lie in a series of polymatroids that describe the system capacity at each moment in time.
This assumption does not hold in our setting where the work rate of each job is defined by the speedup function $s(k)=k^p$.

Hence, standard techniques do not appear to be compatible with our goal of finding the optimal policy with respect to weighted flow time.

\subsection*{Contributions}
In Section \ref{sec:flow}, we provide a complete overview of our results, but below we highlight our main contributions. 
\begin{itemize}[leftmargin=.45cm]
    \item To derive the optimal allocation function, we first develop a new technique in Section \ref{sec:slow} to reduce the dimensionality of the optimization problem. 
This dimensionality reduction leverages two key properties of the optimal policy. 
First, in Section \ref{sec:prelim}, we show that the optimal policy must complete jobs in shortest-job-first order.
Then, in Section \ref{sec:sf}, we prove the scale-free property of the optimal policy which illustrates the optimal substructure of the optimal policy.
These insights reduce the problem of solving $M!$ optimization problems of $O(M^2)$ variables and $O(M)$ constraints to the problem of solving one, unconstrained optimization problem of exactly $M$ variables.  

\item In Section \ref{sec:opt}, we solve our simplified optimization problem to derive the first closed form expression for the optimal allocation of $N$ servers to $M$ jobs which minimizes weighted flow time when weights favor small jobs.  At any moment in time $t$ we define
$$\bm{\theta}^*(t) = (\theta^*_1(t), \theta^*_2(t), \ldots , \theta^*_M),$$
where $\theta^*_i(t)$ denotes the fraction of the $N$ servers allocated to job $i$ at time $t$.
Note that $\theta^*_i(t)$ does not depend on $N$.
Our optimal allocation balances the size-awareness of SRPT and the high efficiency of EQUI.  
We thus refer to our optimal policy as \emph{high efficiency SRPT} (heSRPT) (see Theorem \ref{thm:opt}).
We also provide a closed form expression for the weighted flow time under heSRPT (see Theorem \ref{thm:optt}).

\item 
%In Section \ref{sec:gen}, we show how heSRPT generalizes to the optimal policy for a broader class of objective functions besides mean slowdown, including mean flow time.  We provide the optimal policy with respect to mean flow time in Theorem \ref{thm:optflow}. 
In Section \ref{sec:numerical}, we numerically compare the optimal policies with respect to both mean slowdown and mean flow time to other heuristic policies proposed in the literature and show that our optimal policies significantly outperform these competitors.

\item Finally, Section \ref{sec:numerical:online} turns to the online setting where jobs arrive over time.
We propose an online version of heSRPT called \emph{Adaptive-heSRPT} which uses the allocations from heSRPT to recalculate server allocations on every arrival and departure.
Adaptive-heSRPT significantly outperforms competitor policies from the literature in simulation, often by an order of magnitude.
\end{itemize}

\section{Prior Work}
Despite the prevalence of parallelizable data center workloads, it is not known, in general, how to optimally allocate servers to a set of parallelizable jobs.  
The state-of-the-art in production systems is to let the \emph{user} decide their job's allocation by reserving the resources they desire \cite{verma2015large}, and then to allow the system to pack jobs onto servers \cite{ren2016clairvoyant}.
Users reserve resources greedily, leading to low system efficiency.
Many workloads consist of \emph{malleable} jobs which have the capability to change their degree of parallelism as they run \cite{gupta2014towards,cera2010supporting}.
We seek to improve the status quo by allowing the \emph{system} to choose each job's server allocation at every moment in time.

The closest work to our results is \cite{lin2018model}, which considers jobs which follow a realistic speedup function and have known sizes.
\cite{lin2018model} also allows server allocations to change over time.
\cite{lin2018model} proposes and evaluates heuristic policies such as HELL and KNEE, but they make no theoretical guarantee about the performance of their policies.

Another related work, \cite{berg2018}, assumes that jobs follow a concave speedup function and allows server allocations to change over time.
However, unlike our work, \cite{berg2018} assumes that job sizes are \emph{unknown} and are drawn from an exponential distribution.
\cite{berg2018} concludes that EQUI is the optimal allocation policy.
However, assuming unknown exponentially distributed job sizes is highly pessimistic since this means job sizes are impossible to predict, even as a job ages.

The performance modeling community has considered scheduling in multi-server systems with the goal of minimizing mean response time \cite{harchol2009surprising,Nelson1993PerfromEval,lu2011join,gupta2007analysis,tsitsiklis2011power} or slowdown \cite{wierman2005nearly,hyytia2012minimizing}, but this work has not considered systems where a single job can run on multiple servers.
An exception to this is the work on Fork-Join \cite{ko2004response, wang2019delay} in which a single job is composed of multiple tasks that may run in parallel.
In Fork-Join models, however, the level of parallelism of jobs is fixed and is not chosen by the scheduling policy.

The SPAA/parallel community has studied the problem of allocating servers to jobs which follow arbitrary speedup functions in order to minimize flow time\cite{im2016competitively,Edmonds1999SchedulingIT, edmonds2009scalably,Agrawal:2016:SPJ:2935764.2935782}.
Like our paper, \cite{im2016competitively} considers jobs of known size while \cite{Edmonds1999SchedulingIT, edmonds2009scalably,Agrawal:2016:SPJ:2935764.2935782} consider jobs of unknown size.
These papers use competitive analysis, which assumes that job sizes, arrival times, and even speedup functions are adversarially chosen.  
They conclude that a variant of EQUI is $(1+\epsilon)$-speed $O(1)$-competitive when job sizes are unknown \cite{edmonds2009scalably}.
When job sizes are known, a combination of SRPT and EQUI is $O(\log P)$-competitive, where $P$ is the ratio of the largest job size to the smallest \cite{im2016competitively}.
The SPAA community also considers minimizing mean slowdown with non-parallelizable jobs \cite{anand2012resource,moseley2013complexity,bussema2006greedy}.
Sadly, O(1)-competitive policies for slowdown do not exist in general \cite{bansal2003server}.

Instead of considering speedup functions, the SPAA community often models each job as a DAG of interdependent tasks \cite{blumofe1999scheduling,bampis2014note,bodik2014brief,narlikar1999space}.
This DAG encodes precedence constraints between tasks, and thus implies how parallelizable a job is at every moment in time.
It is not clear how to optimally schedule a \emph{single} DAG job on many servers\cite{chowdhury2013oblivious}.
The problem only gets harder if tasks are allowed to run on multiple servers \cite{du1989complexity,chen2018improved}.
Our hope is that by modeling parallelism using speedup functions, we can address problems that would be intractable in the DAG model.

Our model shares some similarities with coflow scheduling \cite{jahanjou2017asymptotically, chowdhury2014efficient, qiu2015minimizing,shafiee2017brief,khuller2016brief} where one allocates a continuously divisible resource, link bandwidth, to a set of flows to minimize mean flow time.
However, here there is usually no explicit notion of a flow's speedup function.
The most applicable work here is \cite{aalto2011optimal}, which explores the tradeoff between efficiency and opportunistic scheduling in wireless networks.
This work considers mean flow time, not mean slowdown.
Section \ref{sec:gen} shows how our results generalize for other metrics, such as mean flow time, in addition to mean slowdown.

%Our model also shares some similarities with the coflow scheduling problem \cite{jahanjou2017asymptotically, chowdhury2014efficient, qiu2015minimizing,shafiee2017brief,khuller2016brief}.  In coflow scheduling, one must allocate a continuously divisible resource, link bandwidth, to a set of network flows to minimize mean flow time.   Unlike our model, there is no explicit notion of a flow's speedup function here.  Given that this problem is NP-Hard, prior work examines the problem via heuristic policies \cite{chowdhury2014efficient}, and approximation algorithms \cite{qiu2015minimizing,jahanjou2017asymptotically}.

\section{Overview of Our Results}
\label{sec:flow}
%This is equivalent to minimizing the \emph{total} slowdown for a set of $M$ jobs, and thus we consider minimizing total slowdown for the remainder of this section.
%We denote the total slowdown under a policy $P$ to be
%$$\Sl{P} = M \cdot \overline{S}^P.$$
Our goal is to determine the optimal allocation of servers to jobs at every time, $t$, in order to minimize the weighted flow time of a set of $M$ jobs of known size where weights favor small jobs.
We derive a closed form for the optimal allocation function 
$$\bm{\theta}^*(t) = (\theta_1^*(t),\theta_2^*(t), \ldots, \theta_M^*(t))$$
which defines the allocation for each job at any moment in time, $t$, that minimizes weighted flow time.  
We now provide an overview of the main theorems from this paper.

We begin by showing that the optimal policy completes jobs in \emph{Shortest-Job-First} (SJF) order.
The proof of this theorem uses an interchange argument to show that any policy which violates the SJF completion order can be improved.
Because jobs follow a concave speedup function, a standard interchange proof fails.
Our proof requires careful accounting, and interchanges servers between many jobs simultaneously.
This claim is stated in Theorem \ref{thm:sjf}.

\begin{restatable*}[Optimal Completion Order]{thm}{sjf}
    \label{thm:sjf}
    The optimal policy completes jobs in Shortest-Job-First (SJF) order:
    $$M, M-1, M-2, \ldots, 1.$$
\end{restatable*}

\noindent Since jobs are completed in SJF order, we can also conclude that, at time $t$, the jobs left in the system are specifically jobs $1, 2, \ldots, m(t)$.
Theorem \ref{thm:sjf} is proven in Section \ref{sec:prelim}.

Theorem \ref{thm:sjf} does not rely on the specific form of the speedup function -- it holds if $s(k)$ is any increasing, strictly concave function.

Besides completion order, the other key property of the optimal allocation that we exploit is the \emph{scale-free property}.  Our scale-free property states that for any job, $i$, job $i$'s allocation relative to jobs completed after job $i$ (jobs larger than job $i$) is constant throughout job $i$'s lifetime.  This property is stated formally in Theorem \ref{thm:sf}. 

\begin{restatable*}[Scale-free Property]{thm}{sf}
    \label{thm:sf}
    Under the optimal policy, for any job, $i$, there exists a constant $c^*_i$ such that, for all $t<T^*_i$
    $$\frac{\theta_{i}^*(t)}{\sum_{j=1}^{i} \theta_{j}^*(t)} = c^*_i.$$
\end{restatable*}

The scale-free property has an intuitive interpretation.
One can imagine the optimal policy as starting with some optimal allocation function $\bm{\theta}^*(0)$, which gives a $\theta^*_i(0)$ fraction of the servers to job $i$ at time 0.
When job $M$ completes, it will leave behind a set of $\theta^*_M(0) \cdot N$ newly idle servers.
The scale-free property tells us that the new value of the optimal allocation is found by reallocating these idle servers to each job $i<M$ in proportion to $\theta^*_i(0)$.
That is, of the $\theta^*_M(0) \cdot N$ newly available servers, job $i$ should receive 
$$\frac{\theta^*_i(0)}{\sum_{j=1}^{M-1} \theta_{j}^*(0)}\cdot \theta^*_M(0) \cdot N$$
additional servers.

\begin{figure}[ht]
\centering
\includegraphics[width=.9\textwidth]{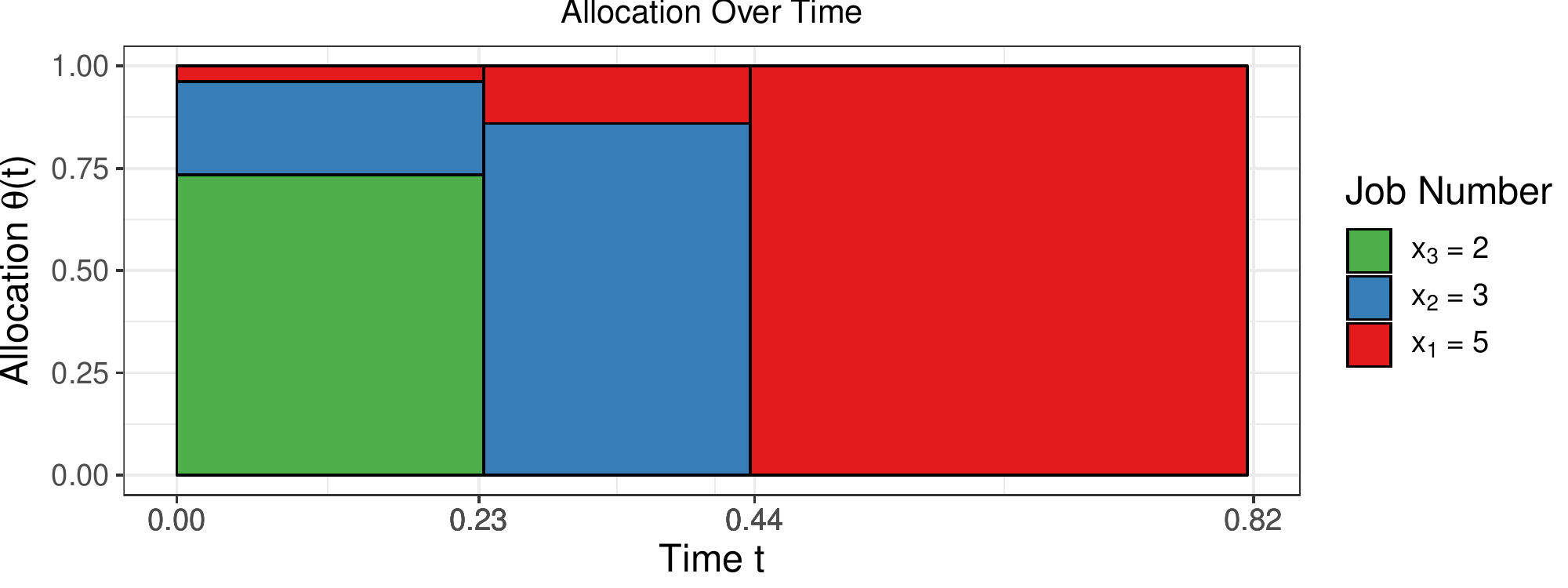}
    \caption{Allocations made by the optimal allocation policy with respect to mean slowdown, heSRPT, completing a set of $M=3$ jobs where the speedup function is $s(k)=k^{.5}$ and $N=100$.
    Job 3 completes at time $t=.23$, job 2 completes at time $t=.44$ and job 1 completes at time $t=.82$.
    Jobs are finished in Shortest Job First order.  Rather than allocating the whole system to the shortest job, heSRPT optimally shares the system between all active jobs.}
\label{fig:3jobs}
\end{figure}

An example of the scale-free property in action can be seen in Figure \ref{fig:3jobs}.
Here, the optimal allocation policy gives all 3 jobs a non-zero allocation a time $t=0$.
When job 3 completes, its servers are reallocated to jobs 1 and 2.
The allocations of job 1 and job 2 both increase, but the ratio of these jobs' allocations remains constant.

The scale-free property provides an optimal substructure that we exploit to reduce the dimensionality of our optimization problem.
For example, consider the case where the optimal allocation function is known for a set of $M-1$ jobs, and we wish to additionally consider adding an $M$th job.
If we first decide how many servers to allocate to job $M$, the scale-free property tells us that stealing these servers from the other $M-1$ jobs in proportion to their existing allocations will produce the optimal policy.
Hence, knowing the optimal policy for a set of $M-1$ jobs reduces the problem of finding the optimal policy for a set of $M$ jobs to a single variable optimization problem.
Theorem \ref{thm:sf} is proven in Section \ref{sec:sf}.

Note that, while we state the scale-free property above for the optimal policy, the scale-free property actually holds for a more a general class of policies.  For any given completion order of the jobs, the policy which minimizes weighted flow time while completing jobs in the given order will obey the scale-free property.

%We can now exploit our knowledge of the optimal completion order and the scale-free property to find a closed form for the optimal policy.
%The first step in deriving the optimal allocation policy is to prove a size-invariant property which says that the optimal allocation to jobs $1, 2, \ldots, m(t)$ at time $t$ depends only on $m(t)$, not on the specific remaining sizes of the these jobs.
%This is a counter-intuitive result because one might imagine that the optimal allocation should be different if two very differently sized jobs are in the system instead of two equally sized jobs.
%The size-invariant property is stated in Theorem \ref{thm:indep}.
%\begin{restatable}[Size-invariant Property]{thm}{indep}
%    \label{thm:indep}
%    Consider any two sets of jobs, $A$ and $B$, each containing $m(t)$ jobs at time $t$.  
%    If $\bm{\theta}^*_A(t)$ and $\bm{\theta}^*_B(t)$ are the optimal allocations to the jobs in sets $A$ and $B$ at time $t$, respectively, we have that
%    $$\bm{\theta}^*_A(t) = \bm{\theta}^*_B(t).$$
%    That is, the allocation function depends only on the number of unfinished jobs in the system.
%\end{restatable}
%
%Theorem \ref{thm:indep} simplifies the computation of the optimal allocation policy, since it allows us to ignore the actual remaining sizes of the jobs in the system.  
%We need only to derive one optimal allocation for each possible value of $m(t)$.  We prove Theorem \ref{thm:indep} in Section \ref{sec:opt}.

We are now finally ready to state Theorem \ref{thm:opt}, which provides the allocation function for the optimal allocation policy which minimizes weighted flow time when weights favors small jobs.
\begin{restatable*}[Optimal Allocation Function]{thm}{opt}
    \label{thm:opt}
    At time $t$, when $m(t)$ jobs remain in the system,
    $$\theta^*_i(t) = \begin{cases} \left(\frac{z(i)}{z(m(t))}\right)^{\frac{1}{1-p}} - \left(\frac{z(i-1)}{z(m(t))}\right)^{\frac{1}{1-p}} & 1\leq i \leq m(t)\\
    \quad 0 & i > m(t)
    \end{cases}$$
    where
    $$z(k)=\sum_{i=1}^{k}w_i.$$
    We refer to the optimal allocation policy which uses $\bm{\theta}^*(t)$ as heSRPT.
\end{restatable*}
Theorem \ref{thm:opt} is proven in Section \ref{sec:opt}.
%Note that the optimal allocation to job $i$ is independent of its remaining size, in accordance with Theorem \ref{thm:indep}.
Given the optimal allocation function $\bm{\theta}^*(t)$, we can also explicitly compute the optimal weighted flow time for any set of $M$ jobs.
This is stated in Theorem \ref{thm:optt}.
\begin{restatable*}[Optimal Weighted Flow Time]{thm}{optt}
    \label{thm:optt}
    Given a set of $M$ jobs of size $x_1 \geq x_2 \geq \ldots \geq x_M$, the weighted flow time, $\Sl{*}$, under the optimal allocation policy $\bm{\theta}^*(t)$ is given by
    $$\Sl{*}=\frac{1}{s(N)}\sum_{k=1}^M  x_{k}\cdot\bigl[z(k)^{\frac{1}{1-p}} - z(k-1)^{\frac{1}{1-p}}\bigr]^{1-p}$$
    where
    $$z(k)=\sum_{i=1}^{k}w_i.$$
\end{restatable*}

Note that the optimal allocation policy biases its allocations towards shorter jobs, but does not give strict priority to these jobs in order to maintain the overall efficiency of the system.  That is,
$$0<\bm{\theta}^*_1(t) <\bm{\theta}^*_2(t) < \ldots < \bm{\theta}^*_{m(t)}(t).$$
This is illustrated in Figure \ref{fig:3jobs}, where the smallest remaining jobs always get the largest allocations under the optimal policy.
We refer to the optimal policy derived in Theorem \ref{thm:opt} as \emph{High Efficiency Shortest-Remaining-Processing-Time} or \emph{heSRPT}.

Section \ref{sec:gen} discusses how to apply Theorems \ref{thm:opt} and \ref{thm:optt} in order to optimize both the mean slowdown and mean flow time metrics.
These applications are summarized in Corollary \ref{thm:optflow}.
\begin{restatable*}[Optimal Mean Slowdown and Mean Flow Time]{cor}{optflow}
    \label{thm:optflow}
If we define
    $$w_i=\frac{1}{x_i/s(N)} \quad\mbox{and thus}\quad z(k)=\sum_{i=1}^{k}\frac{1}{x_i/s(N)},$$
    Theorem \ref{thm:opt} yields the optimal policy with respect to mean slowdown and Theorem \ref{thm:optt} yields the optimal total slowdown.

If we define
    $$w_i=1 \quad\mbox{and thus}\quad z(k)=k,$$
    Theorem \ref{thm:opt} yields the optimal policy with respect to mean flow time and Theorem \ref{thm:optt} yields the optimal total flow time.
\end{restatable*}

The remainder of Section \ref{sec:eval} is devoted to the development of Adaptive-heSRPT, an online version of heSRPT.

\section{Minimizing Weighted Flow Time}
\label{sec:slow}
\subsection{The Optimal Completion Order}\label{sec:prelim}
To determine the optimal completion order of jobs, we first show that the optimal allocation function remains constant between job departures.
This implies that the optimal allocation has $M$ decision points where new allocations must be determined.
\begin{restatable}{thm}{constant}
    \label{thm:constant}
    Consider any two times $t_1$ and $t_2$ where, WLOG, $t_1<t_2$.  Let $m^*(t)$ denote the number of jobs in the system at time $t$ under the optimal policy.  If $m^*(t_1)=m^*(t_2)$ then
    $$\bm{\theta}^*(t_1) = \bm{\theta}^*(t_2).$$
\end{restatable}

\begin{proof}
    This proof is straightforward, and leverages the concavity of the speedup function, $s(k)$.  See \ref{sec:constant}.
\end{proof}
For the rest of the paper, we will therefore only consider allocation functions which change exclusively at departure times.

We can now show that the optimal policy will complete jobs in Shortest-Job-First order.

\sjf

\begin{proof}
Assume the contrary, that the optimal policy does not follow the SJF completion order.  This means there exist two jobs $\alpha$ and $\beta$ such that $x_\alpha < x_\beta$ but $T^*_\beta < T^*_\alpha$.  Note that any number of completions may occur between job $\alpha$ and job $\beta$.

If the allocation to job $\alpha$ had always been at least as large as the allocation to job $\beta$, it is easy to see that job $\alpha$ would have finished first.  Hence, we can identify some intervals of time where: (i) the allocations to job $\alpha$ and job $\beta$ are constant because there are no departures and (ii) the allocation to job $\beta$ is higher than the allocation to job $\alpha$.  Let $I$ be the smallest set of disjoint intervals such that for every $[t_1, t_2) = i \in I$
\begin{align}
\theta^*_\alpha(t) = \theta^*_\alpha(t') &\mbox{ and } \theta^*_\beta(t) = \theta^*_\beta(t') \qquad \forall t, t' \in [t_1, t_2)\\
\theta^*_\beta(t_1) &> \theta^*_\alpha(t_1).
\end{align}
Note also that on the interval $[T^*_\beta, T^*_\alpha)$, job $\alpha$ receives some positive allocation while job $\beta$ receives no allocation because it is complete.

On some subset of these intervals we will perform an interchange which will reduce the weighted flow time of the jobs, producing a contradiction.  Let $P$ be the policy resulting from this interchange.
Under $P$, we will fully exchange $\beta$ and $\alpha$'s allocations on the interval $[T^*_\beta, T^*_\alpha)$.  
That is, for any $t\in [T^*_\beta, T^*_\alpha)$, $\theta^P_\beta(t) = \theta^*_\alpha(t)$ and $\theta^P_\alpha(t)=0$.
To offset this interchange we will reallocate servers from $\beta$ to $\alpha$ on some subset of the intervals in $I$.
We will argue that, after this interchange, $T^P_\beta = T^*_\alpha$ and $T^P_\alpha < T^*_\beta$, leading to a reduction in weighted flow time.  

Let $W^P_i(t)$ be the total amount of work done by policy $P$ on job $i$ by time $t$.  Let $W^P_i(t_1,t_2)$ be the amount of work done by policy $P$ on job $i$ on the interval $[t_1,t_2)$.  That is, 
$$W^P_i(t_1,t_2) = W^P_i(t_2)-W^P_i(t_1).$$
Let $\gamma=W^*_\alpha(T^*_\beta,T^*_\alpha)$.  Since job $\alpha$ finished at exactly time $T^*_\alpha$, we know that $W^*_\alpha(T^*_\beta) = x_\alpha - \gamma$.
Similarly, we know that 
$$W^*_\beta(T^*_\beta) = x_\beta.$$
We can see that policy $P$ has the capacity to do up to $\gamma$ work on job $\beta$ on the interval $[T^*_\beta, T^*_\alpha)$. Specifically, if
$$W^P_\beta(T^*_\beta) = x_\beta - \gamma$$
then $T^P_\beta = T^*_\alpha$.

We will now reallocate servers from job $\beta$ to job $\alpha$ on some of the intervals in $I$ until $W^P_\beta(T^*_\beta) = x_\beta - \gamma$.  On every such interval $i=[t_1, t_2) \in I$ we will choose 
$$0 \leq \delta_i \leq  \theta^*_\beta(t_1) -  \theta^*_\alpha(t_1)$$ 
and let 
$$\theta^P_\beta(t_1) = \theta^*_\beta(t_1) - \delta_i \qquad \mbox{and} \qquad \theta^P_\alpha(t_1) = \theta^*_\alpha(t_1) + \delta_i.$$
Decreasing job $\beta$'s allocation on this interval has a well defined effect on $W^P_\beta(T^*_\beta)$.  Namely, decreasing an allocation of a $\theta^*_\beta(t_1)$ fraction of the servers by $\delta_i$ decreases  $W^P_\beta(T^*_\beta)$ by 
$$(t_2-t_1)(s(\theta^*_\beta(t_1) \cdot N)-s((\theta^*_\beta(t_1)-\delta_i)\cdot N)).$$
Note that, for any interval $i\in I$, $W^P_\beta(T^*_\beta)$ is a continuous, decreasing function of $\delta_i$.

We will now iteratively apply the following interchange algorithm.  Initialize $\delta_i = 0$, $\forall i\in I$.  For each interval $i \in I$, try setting $\delta_i$ to be $\theta^*_\beta(i) - \theta^*_\alpha(i)$.  If $W^P_\beta(T^*_\beta)$ is greater than $x_\beta - \gamma$ after this interchange, proceed to the next interval.  If $W^P_\beta(T^*_\beta) < x_\beta - \gamma$ using this value of $\delta_i$, there must exist some $0<\delta_i < \theta^*_\beta(i) - \theta^*_\alpha(i)$ such that $W^P_\beta(T^*_\beta) = x_\beta - \gamma$ by the intermediate value theorem.  Select this value of $\delta_i$ and terminate the algorithm.

To see this algorithm terminates, consider what would happen if 
$$\delta_i = \theta^*_\beta(i) - \theta^*_\alpha(i) \qquad \forall i \in I.$$
In this case, 
$$\theta^P_\beta(t) \leq \theta^*_\alpha(t) \qquad \forall t \leq T^*_\beta$$
which would imply that 
$$W^P_\beta(T^*_\beta) \leq W^*_\alpha(T^*_\beta) = x_\alpha - \gamma < x_\beta - \gamma.$$
Hence the algorithm will terminate before running out of intervals in $I$.  By construction, we know that $T^P_\beta=T^*_\alpha$.

We now show that, after performing the interchange, $T^P_\alpha < T^*_\beta$.  To see this, consider the total amount of work done on any interval in $i=[t_1,t_2)\in I$ before and after the interchange.  Because only the allocations to $\alpha$ and $\beta$ have changed, it is clear that
    \begin{align*}\sum_{j=1}^{m(t_1)} W^P_j(t_1,t_2)-W^*_j(t_1,t_2) = (W^P_\alpha(t_1,t_2) + W^P_\beta(t_1,t_2)) &- (W^*_\alpha(t_1,t_2) + W^*_\beta(t_1,t_2)).
\end{align*}
%Furthermore, we know that there exists some $\delta_i$ such that
%$$\theta^P_\beta = \theta^*_\beta - \delta_i$$
%and
%$$\theta^P_\alpha = \theta^*_\alpha + \delta_i.$$
By the strict concavity of $s$, we know that the second derivative of $s$ is negative.
    As illustrated in Figure \ref{fig:sjf}, and because $\theta^*_\beta(t_1) - \delta_i\geq \theta^*_\alpha(t_1)$ by construction, we know that
%    {\small
    \begin{align}
        s((\theta^*_\alpha(t_1) + \delta_i) N) - s(\theta^*_\alpha(t_1) N) > s(\theta^*_\beta(t_1)  N) - s((\theta^*_\beta(t_1) - \delta_i) N)&\label{eq:conc}.
    \end{align}
%    }
\begin{figure}[ht]
\centering
\includegraphics[width=.65\textwidth]{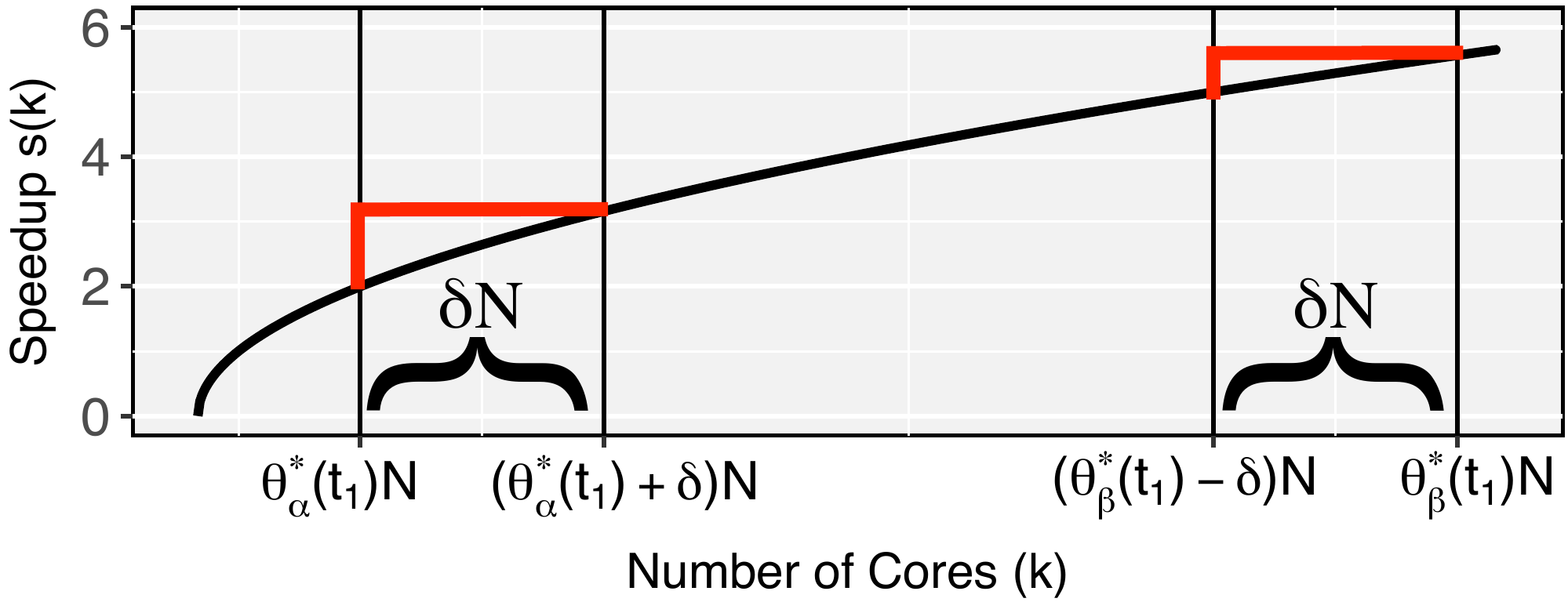}
    \caption{
        An illustration of \eqref{eq:conc}.  
        Because the speedup function, $s$, is concave, its second derivative is negative.  
        Hence, $s$ increases less on the interval $[(\theta^*_\beta(t_1) - \delta)N, \theta^*_\beta(t_1)N]$ than on the interval $[\theta^*_\alpha(t_1) N,(\theta^*_\alpha(t_1)+\delta)N]$. 
    }
\label{fig:sjf}
\end{figure}

Factoring out $N$ gives
\begin{align*}
s(\theta^*_\alpha(t_1) + \delta_i) - s(\theta^*_\alpha(t_1)) &> s(\theta^*_\beta(t_1)) - s(\theta^*_\beta(t_1) - \delta_i)
\end{align*}
and thus
%    {\small
\begin{align*}
    \bigl(s(\theta^*_\alpha(t_1) + \delta_i) +s(\theta^*_\beta(t_1) - \delta_i)\bigr) -
    \bigl(s(\theta^*_\alpha(t_1)) + s(\theta^*_\beta(t_1))\bigr) > 0.
\end{align*}
%    }
Multiplying both sides by $(t_2-t_1)$, we see that
%\begin{align*}
%    {\small
%    (t_2-t_1)[(s(\theta^*_\alpha(t_1) + \delta_i) +s(\theta^*_\beta(t_1) -\delta_i)) -
%    (s(\theta^*_\alpha(t_1)) + s(\theta^*_\beta(t_1)))] > 0
%    }
%\end{align*}
%which can be rewritten as
\begin{align*}
(W^P_\alpha(t_1,t_2) + W^P_\beta(t_1,t_2)) &- (W^*_\alpha(t_1,t_2) + W^*_\beta(t_1,t_2)) > 0
\end{align*}
since both allocations are constant on the interval $[t_1,t_2)$.
That is, the total amount of work done on interval $i$ has increased.  Since this holds for all $i\in I$ we have that 

%{\small
\begin{align}
    \sum_{[t_1,t_2) \in I} (W^P_\alpha(t_1,t_2) + W^P_\beta(t_1,t_2)) -
    (W^*_\alpha(t_1,t_2) + W^*_\beta(t_1,t_2)) &> 0\nonumber\\
    \sum_{[t_1,t_2) \in I} W^P_\alpha(t_1,t_2) - W^*_\alpha(t_1,t_2) - 
    \sum_{[t_1,t_2) \in I} W^*_\beta(t_1,t_2) -  W^P_\beta(t_1,t_2) &> 0.
\label{eq:ineq}
\end{align}
%}
Again, by construction, we have decreased the amount of work done on job $\beta$ during the intervals in $I$ by exactly $\gamma$.  That is,
\begin{align*}
%W^*_\beta(T^*_\beta) - W^P_\beta(T^*_\beta) &= x_\beta - (x_{\beta} - \gamma)\\
 \sum_{[t_1,t_2) \in I} W^*_\beta(t_1,t_2) -  W^P_\beta(t_1,t_2) &= \gamma
 \end{align*}
 and hence by \eqref{eq:ineq}
 $$\sum_{[t_1,t_2) \in I}W^P_\alpha(t_1,t_2) - W^*_\alpha(t_1,t_2) > \gamma.$$
 Recall that $W^*_\alpha(T^*_\beta) = x_\alpha - \gamma$.  Thus
 \begin{align*}
 W^P_\alpha(T^*_\beta) - W^*_\alpha(T^*_\beta) &= \sum_{[t_1,t_2) \in I}W^P_\alpha(t_1,t_2) - W^*_\alpha(t_1,t_2)\\
 W^P_\alpha(T^*_\beta) - (x_\alpha - \gamma) &> \gamma\\
 W^P_\alpha(T^*_\beta) &> x_\alpha.
 \end{align*}
This implies that, under policy $P$, job $\alpha$ finishes before time  $T^*_\beta$.

We can thus conclude that after the interchange, policy $P$ completes job $\beta$ at exactly $T^*_\alpha$ and policy $P$ completes job $\alpha$ at some time before $T^*_\beta$.
All other jobs are unchanged by this interchange, and hence their flow times remain the same.
Hence, it suffices to show that 
$$w_\alpha T^P_\alpha + w_\beta T^P_\beta < w_\alpha T^*_\alpha + w_\beta T^*_\beta.$$
Because weights favor small jobs we have that $w_\beta < w_\alpha$.
Furthermore, we have assumed that $T^*_\alpha > T^*_\beta$.
Hence,
$$w_\beta \left(T^*_\alpha - T^*_\beta\right) < w_\alpha \left(T^*_\alpha - T^*_\beta\right)$$
$$w_\beta T^*_\alpha + w_\alpha T^*_\beta < w_\alpha T^*_\alpha + w_\beta T^*_\beta,$$
and thus by construction
$$w_\beta T^P_\beta + w_\alpha T^P_\alpha < w_\beta T^*_\alpha + w_\alpha T^*_\beta < w_\alpha T^*_\alpha + w_\beta T^*_\beta.$$
This implies that the policy $P$ has a lower weighted flow time than the optimal policy, a contradiction.
\end{proof}
\begin{restatable}{definition}{sjfimpact}
    \label{rem:sjf}
    \noindent It will be useful to consider the rate at which weighted flow time is accrued in between departures.
    Because the optimal completion order is SJF, we know that job $k$ will complete directly after job $k+1$, and that during the interval $[T^*_{k+1}, T^*_{k})$, jobs 1 through $k$ will be present in the system.
    Hence, we define $z(k)$ to be the rate at which total flow time is accrued during the interval $[T^*_{k+1}, T^*_{k})$, where
    $$z(k) = \sum_{i=1}^k w_i.$$
\end{restatable}

\subsection{The Scale-Free Property}
\label{sec:sf}
The goal of this section is to characterize some optimal substructure of the optimal policy that will allow us to reduce the search space for the optimal policy.
Hence, we now prove an interesting property of the optimal policy which we call the scale-free property.
We will first need a preliminary lemma.
\begin{restatable}{lem}{scaling}
    \label{lem:scaling}
    Consider an allocation function $\bm{\theta}(t)$ which, on some time interval $[0,T]$ leaves $\beta$ fraction of the system unused.
    That is,
    $$\sum_{i=1}^{m(t)} \theta_i(t) = 1-\beta \qquad \forall t\in[0,T].$$
    The total work done on any job $i$ by time $T$ under $\bm{\theta}(t)$ is equivalent to the total work done on job $i$ by time $T$ under an allocation function $\bm{\theta}'(t)$ where
    $$\bm{\theta}'(t) = \frac{\bm{\theta}(t)}{1-\beta} \qquad \forall t\in[0,T]$$
    in a system that runs at $\left(1-\beta\right)^p$ times the speed of the original system (which runs at rate 1).
\end{restatable}
\begin{proof}
    Is straightforward, see \ref{sec:morescaling}.
\end{proof}
\noindent 
Using Lemma \ref{lem:scaling} we can characterize the optimal policy.  Theorem \ref{thm:sf} states that a job's allocation relative to the jobs larger than it will remain constant for the job's entire lifetime.
\sf
\begin{proof}
    We will prove this statement by induction on the overall number of jobs, $M$. 
    First, note that the statement is trivially true when $M=1$. 
    It remains to show that if the theorem holds for $M=k$, then it also holds for $M=k+1$.  
    
    Let $M=k+1$ and let $T^*_{i}$ denote the finishing time of job $i$ under the optimal policy.
    Recall that the optimal policy finishes jobs according to the SJF completion order, so $T^*_{{i+1}} \leq T^*_{i}$.  Consider a system which optimally processes $k$ jobs, which WLOG are jobs $1, 2, \ldots, {M-1}$.  We will now ask this system to process an additional job, job ${M}$.
    From the perspective of the original $k$ jobs, there will be some constant portion of the system, $\theta^*_{M}$, used to process job $M$ on the time interval $[0,T^*_{M}]$.  The remaining $1-\theta^*_{M}$ fraction of the system will be available during this time period.  Just after time $T^*_{M}$, there will be at most $k$ jobs in the system, and hence by the inductive hypothesis the optimal policy will obey the scale-free property on the interval $(T^*_{{M}},T^*_{{1}}]$.
    
    Consider the problem of minimizing the weighted flow time of the $M$ jobs \emph{given any fixed value} of $\theta^*_{M}$ such that the SJF completion order is obeyed.  We can write the optimal weighted flow time of the $M$ jobs, $\Sl{*}$, as
    $$\Sl{*}=w_M T^*_M + \sum_{j=1}^{M-1} w_j T^*_j$$
    where $w_M T^*_M$ is a constant.  Clearly, optimizing weighted flow time in this case is equivalent to optimizing the weighted flow time for $M-1=k$ jobs with the added constraint that $\theta^*_{M}$ is unavailable (and hence ``unused'' from the perspective of jobs ${M-1}$ through ${1}$) during the interval $[0, T^*_{M}]$.  By Lemma \ref{lem:scaling}, this is equivalent to having a system that runs at a fraction $\left(1-\theta^*_{M}\right)^p$ of the speed of a normal system during the interval $[0, T^*_{M}]$.

    Thus, for some $d>1$, we will consider the problem of optimizing weighted flow time for a set of $k$ jobs in a system that runs at a speed $\frac{1}{d}$ times as fast during the interval $[0,T^*_{M}]$.

    Let $\Sl{*}[x_{{M-1}}, \ldots, x_{{1}}]$ be the optimal weighted flow time of $k$ jobs of size $x_{{M-1}} \ldots x_{1}$.  Let $\Sl{s}[x_{{M-1}}, \ldots, x_{{1}}]$ be the weighted flow time of these jobs in a \emph{slow system} which always runs $\frac{1}{d}$ times as fast as a normal system.

If we let $T^s_i$ be the finishing time of job $i$ in the slow system, it is easy to see that 
$$T^*_i = \frac{T^s_i}{d}$$
since we can just factor out a $d$ from the expression for the completion time of every job in the slow system.  
Hence, we see that
    $$\Sl{*}[x_{{M-1}}, \ldots, x_{{1}}] = \frac{\Sl{s}[x_{{M-1}}, \ldots, x_{{1}}]}{d}$$
by the same reasoning.  Clearly, then, the optimal allocation function in the slow system, $\bm{\theta}^s(t)$, is equal to $\bm{\theta}^*(t)$ at the respective departure times of each job.  That is,
$$\bm{\theta}^*(T^*_j) = \bm{\theta}^s(T^s_j) \qquad \forall 1 \leq j \leq M-1.$$

We will now consider a \emph{mixed} system which is ``slow'' for some interval $[0,T^*_{M}]$ that ends before $T^s_{{M-1}}$, and then runs at normal speed after time $T^*_{M}$.
    Let $\Sl{Z}[x_{{M-1}}, \ldots, x_{{1}}]$ denote the weighted flow time in this mixed system and let $\bm{\theta}^Z(t)$ denote the optimal allocation function in the mixed system.
We can write
\begin{align*}
    \Sl{Z}[x_{{M-1}}, \ldots, x_{{1}}] &= z(k)\cdot T^*_{M}\\
    &+ \Sl{*}[x_{{M-1}}-\frac{s(\theta_{{M-1}}^Z)T^*_{M}}{d}, \ldots x_{{1}} - \frac{s(\theta_{{1}}^Z)T^*_{M}}{d}].
\end{align*}
Similarly we can write

    \begin{align*}\Sl{s}[x_{{M-1}}, \ldots, x_{{1}}] &= z(k)\cdot T^*_{M}\\
        &+ \Sl{s}[x_{{M-1}}-\frac{s(\theta_{{M-1}}^s)T^*_{M}}{d}, \ldots x_{{1}} - \frac{s(\theta_{{1}}^s)T^*_{M}}{d}].
    \end{align*}

    Let $T^Z_{i}$ be the finishing time of job $i$ in the mixed system under $\bm{\theta}^Z(t)$.
    Since $z(k)\cdot T^*_{M}$ is a constant not dependent on the allocation function, we can see that the optimal allocation function in the mixed system will make the same allocation decisions as the optimal allocation function in the slow system at the corresponding departure times in each system.  That is,
    $$\bm{\theta}^*(T^*_j) = \bm{\theta}^s(T^s_j) = \bm{\theta}^Z(T^Z_j) \qquad \forall 1 \leq j \leq M-1.$$
    
    By the inductive hypothesis, the optimal allocation function in the slow system obeys the scale-free property. Hence, $\bm{\theta}^Z(t)$ also obeys the scale-free property for this set of $k$ jobs given any fixed value of $\theta^*_{M}$.

We now apply Lemma \ref{lem:scaling} again, multiplying $\bm{\theta}^Z(t)$ by $1-\theta^*_{M}$ on the interval $[0,T^*_{M}]$ to recover an allocation policy with $\theta^*_{M}$ unused servers on $[0,T^*_{M}]$.
We call the resulting policy $\bm{\theta}^P(t)$.
    By Lemma \ref{lem:scaling}, jobs $M-1, \ldots, 1$ have the same residual sizes at time $T^*_M$ in a mixed system under $\bm{\theta}^Z(t)$ as they do in a constant speed system under $\bm{\theta}^P(t)$.
Hence, because the mixed system under $\bm{\theta}^Z(t)$ and the constant speed system under $\bm{\theta}^P(t)$ are identical after time $T^*_M$, the weighted flow time in these two systems is the same.
Therefore, if $\bm{\theta}^Z(t)$ is optimal in the mixed system, $\bm{\theta}^P(t)$ minimizes the weighted flow time of jobs $M-1, \ldots, 1$ for any fixed value of $\theta^*_{M}$ in a normal speed system.
Since $\bm{\theta}^Z(t)$ obeys the scale-free property, $\bm{\theta}^P(t)$ obeys the scale-free property for jobs $1,\ldots,M-1$.
Hence, for a set of $M=k+1$ jobs, given any allocation $\theta^*_{M}$ to job $M$ on the interval $[0,T^*_{M}]$, the optimal allocations to the remaining $M-1$ jobs obey the scale-free property.
    Finally, the scale-free property is trivially satisfied for job $M$ by allocating any fixed $\theta^*_{M}$ to job $M$ on the interval $[0,T^*_{M}]$ and setting $c^*_i=\theta^*_{M}$.
The optimal allocation function for processing the $M=k+1$ jobs therefore obeys the scale-free property.
This completes the proof by induction.
\end{proof}
\begin{restatable}{definition}{sfimpact}
    \label{rem:sf}
    \noindent The scale-free property tells us that, under the optimal policy, job $i$'s allocation relative to the jobs completed after it is a constant, $c^*_i$.  Note that the jobs completed after job $i$ are precisely the jobs with an initial size of at least $x_i$, since the optimal policy follows the SJF completion order.  It will be useful to define an {\bf optimal scale-free constant}, $\omega^*_i$, for every job $i$, where for any $t<T^*_{i}$
    $$\omega^*_{i} = \frac{1}{c^*_i} - 1 = \frac{\sum_{j=1}^{i-1} \theta^*_{j}(t)}{\theta^*_{i}(t)}.$$
    Note that we define $\omega^*_1 =0$.  Let $\bm{\omega}^*=(\omega^*_1, \omega^*_2, \ldots, \omega^*_M)$ denote the optimal scale-free constants corresponding to each job.
\end{restatable}

\subsection{Finding the Optimal Allocation Function} 
\label{sec:opt}
We will now make use of the scale-free property and our knowledge of the optimal completion order to find the optimal allocation function.
We will consider the weighted flow time under any policy, $P$, which obeys the scale-free property and follows the SJF completion order.
In Lemma \ref{lem:transform}, we derive an expression for the weighted flow time under the policy $P$ as a function of the \emph{scale-free constants} $\bm{\omega}^P$.
In Theorem \ref{thm:opt-omega} we then minimize this expression to find the optimal scale-free constants and the optimal allocation function.

\begin{restatable}{lem}{transform}
    \label{lem:transform}
Consider a policy $P$ which obeys the scale-free property and which completes jobs in shortest-job-first order.  We define
    $$\omega^P_{i} = \frac{\sum_{j=1}^{i-1} \theta^P_{j}(t)}{\theta^P_{i}(t)} \qquad \forall 1<i\leq M,\ 0\leq t<T^P_{i}$$
    and $\omega^P_1 = 0$. We can then write the weighted flow time under policy $P$ as a function of $\bm{\omega}^P = (\omega^P_1, \omega^P_2, \ldots, \omega^P_M) $ as follows
\begin{eqnarray*}
    \Sl{P}(\bm{\omega}^P)
    &=& \frac{1}{s(N)}\sum_{k=1}^M  x_{k} \cdot \bigl[z(k)  s(1 + \omega^P_{k})  - z(k-1) s(\omega^P_{k})\bigr]\\
\end{eqnarray*}
\end{restatable}
\begin{proof}
To analyze the total weighted flow time under $P$ we will relate $P$ to a simpler policy, $P'$, which is much easier to analyze.
We define $P'$ to be
    $$\theta^{P'}_i = \theta^{P'}_i(t) = \theta^P_i(0) \qquad \forall 1 \leq i \leq M.$$
Importantly, each job receives some initial optimal allocation at time 0 which does not change over time under $P'$.
Since allocations under $P'$ are constant we have that
\begin{eqnarray*}
T^{P'}_k  &=& \frac{x_k}{s(\theta^{P'}_k)} .
\end{eqnarray*}
We can now derive equations that relate $T^P_k$ to $T^{P'}_k$.

By Theorem \ref{thm:sf}, during the interval $[T^{P'}_{{k+1}}, T^{P'}_{{k}}]$,  
$$\omega^P_i = \omega^{P'}_i \qquad \forall 1 \leq i \leq k.$$
Note that a fraction of the system $\sum_{i=k+1}^M \theta^{P'}_{i}$ is unused during this interval, and hence by Lemma \ref{lem:scaling}, we have that
\begin{eqnarray*}
T^{P'}_{k} - T^{P'}_{{k+1}} &=& \frac{T^P_{k} - T^P_{{k+1}}}{s(\theta^{P'}_{1} + \cdots + \theta^{P'}_{k})}. \\
\end{eqnarray*}
    Let $\alpha_k = \frac{1}{s(\theta^{P'}_{1} + \cdots + \theta^{P'}_{k})}$ be the scaling factor during this interval.

    If we define $x_{{M+1}} = 0$ and $T^P_{{M+1}}=0$, we can express the weighted flow time under policy $P$, $\Sl{P}$, as
\begin{align*}
    \Sl{P} &= z(M) T^P_{M} + z(M-1) (T^P_{{M-1}} - T^P_{M}) + \cdots +z(2) (T^P_{{2}} - T^P_{{3}})  + z(1)(T^P_{1} - T^P_{{2}})\\
    &= \sum_{k=1}^M  z(k) (T^P_{k} - T^P_{{k+1}})  \\
    &= \sum_{k=1}^M  z(k) \frac{T^{P'}_{k} - T^{P'}_{{k+1}}}{\alpha_k} .
\end{align*}
    We can now further expand this expression in terms of the job sizes, using the fact that $s(ab)=s(a)\cdot s(b)$, as follows:
\begin{align}
    \Sl{P} =& \sum_{k=1}^M  z(k) \frac{\frac{x_{k}}{s(\theta^{P'}_{{k}}N)} - \frac{x_{{k+1}}}{s(\theta^{P'}_{{k+1}}N)}}{\alpha_k} \nonumber\\
    =& \frac{1}{s(N)}\sum_{k=1}^M  z(k)\cdot \bigl[x_{k} s(1+\omega^{P'}_k) -x_{{k+1}} s(\omega^{P'}_{k+1})\bigr]  \nonumber \\
    =&\frac{1}{s(N)}\sum_{k=1}^M  x_{k}\cdot \bigl[z(k)  s(1 + \omega^P_{k})  - z(k-1) s(\omega^P_{k})\bigr] \label{eq:slowdown}
\end{align}
as desired.
\end{proof}

We now have an expression for the weighted flow time of any policy $P$ which obeys the scale-free property and completes jobs in SJF order.
Since the optimal policy obeys these properties, the choice of $P$ which minimizes the above expression for weighted flow time must be the optimal policy.
In Theorem \ref{thm:opt-omega} we find a closed form expression for the optimal scale-free constants.

\begin{restatable}[Optimal Scale-Free Constants]{thm}{optomega}
    \label{thm:opt-omega}
    The optimal scale-free constants are given by the expression
    $$\omega_{k}^* = \frac{1}{\left(\frac{z(k)}{z(k-1)}\right)^{\frac{1}{1-p}} - 1} \qquad \forall 1 < k \leq M.$$
   
\end{restatable}

\begin{proof}
    Consider a policy $P$ which obeys the scale-free property and completes jobs in SJF order.
    Let $\Sl{P}(\bm{\omega}^P)$ be the expression for weighted flow time for $P$ from Lemma \ref{lem:transform}.
    Our goal is to find a closed form expression for $\bm{\omega}^*$, the scale-free constants which minimize the above expression for weighted flow time.

A sufficient condition for finding $\bm{\omega}^*$ is that a policy completes jobs in SJF order \emph{and} satisfies the following first-order conditions:
\begin{eqnarray*}
    \frac{\partial \Sl{P}}{\partial \omega^P_{k}} = 0 \qquad \forall 1 \leq k \leq M
\end{eqnarray*} 

The second-order conditions are satisfied trivially.
Note that solely satisfying the first order conditions is insufficient, because the resulting solution is not guaranteed to respect the SJF completion order.
Luckily, we can show that any solution which satisfies the first-order conditions also completes jobs in SJF order.
To see this, we will begin by finding the allocation function, $\bm{\Theta}^{P}(t)$, which satisfies the first-order conditions.
%We will see that under $\bm{\Theta}^{P}(t)$, jobs are completed in SJF order.

    The first order conditions, obtained by differentiating \eqref{eq:slowdown}, give
    $$z(k)s'(1+\omega^P_{k}) - z(k-1) s'(\omega^P_{k}) = 0 \qquad \forall 1 \leq k \leq M$$
and hence
    $$\omega_{k}^P = \frac{1}{\left(\frac{z(k)}{z(k-1)}\right)^{\frac{1}{1-p}} - 1} \qquad \forall 1 < k \leq M$$

We can show that the values of $\Theta^{P}_{k}(t)$ are increasing in $k$ (see \ref{sec:addl}).  
That is, smaller jobs always have larger allocations than larger jobs under $\bm{\Theta}^{P}(t)$.
This implies that $\bm{\Theta}^{P}(t)$ follows the SJF completion order, since a larger job cannot complete before a smaller job unless it receives a larger allocation for some period of time.
$\bm{\Theta}^{P}(t)$ therefore satisfies the sufficient condition for optimality.
We thus have found a closed form expression for the optimal scale free constants.
That is,
    $$\omega_{k}^* = \frac{1}{\left(\frac{z(k)}{z(k-1)}\right)^{\frac{1}{1-p}} - 1} \qquad \forall 1 < k \leq M$$
as desired.
\end{proof}

%This computation begins with the interesting observation that the optimal allocation function does not directly depend on the sizes of the $M$ jobs.  This is stated in Theorem \ref{thm:indep}.
%\indep*
%\begin{proof}
%Recall from the proof of Theorem \ref{thm:opt-omega} that the optimal allocation function satisfies the following first order conditions
%    $$z(k)s'(1+\omega^P_{k}) - z(k-1) s'(\omega^P_{k}) = 0 \qquad \forall 1 \leq k \leq M$$
%and always completes jobs in SJF order.
%    Note that these conditions do not explicitly depend on any job sizes.  Hence, while the value of the optimal allocation function may depend on \emph{how many} jobs are in the system, given any two sets of jobs, $A$ and $B$, which consist of $m(t)$ jobs at time $t$, the optimal allocation function for set $A$ will be equal to the optimal allocation function for set $B$ at time $t$.
%\end{proof}
\noindent We now derive the optimal allocation function in Theorem \ref{thm:opt}.

\opt
\begin{proof}
    We can now solve a system of equations to derive the optimal allocation function.
    Consider a time, $t$, when there are $m(t)$ jobs in the system.
    Since the optimal completion order is SJF, we know that the jobs in the system are specifically jobs $1,2, \ldots, m(t)$.
    We know that the allocation to jobs $m(t) + 1, \ldots, M$ is 0, since these jobs have been completed.  Hence, we have that
    $$\theta^*_1 + \theta^*_2 + \ldots + \theta^*_{m(t)} = 1.$$
    Furthermore we have $m(t) - 1$ constraints provided by the expressions for $\omega^*_2, \omega^*_3, \ldots, \omega^*_{m(t)}$.
\begin{eqnarray*}
    \omega^*_2 =  \frac{\theta^*_1(t)}{\theta^*_2(t)} ,  \omega^*_3 = \frac{\theta^*_2(t) + \theta^*_1(t)}{\theta^*_3(t)},
    \cdots,\omega^*_{m(t)} = \frac{\sum_{i=1}^{m(t)-1}\theta^*_{i}(t)}{\theta^*_{m(t)}(t)}.
\end{eqnarray*}
These can be written as
\begin{eqnarray*}
    \theta^*_1(t) &=& \omega^*_2 \theta^*_2(t) \\
    \theta^*_1(t) + \theta^*_2(t) &=& \omega^*_3 \theta^*_3(t) \\
\cdots\\
    \theta^*_1(t) + \theta^*_2(t) + \cdots + \theta^*_{m(t)-1}(t)  &=& \omega^*_{m(t)} \theta_{m(t)} \\
    \theta^*_1 + \theta^*_2 + \ldots + \theta^*_{m(t)} &=& 1\\
\end{eqnarray*}

and then rearranged as
\begin{eqnarray*}
    \theta^*_{m(t)}(t) &=& \frac{1}{1 + \omega^*_{m(t)}} \\
    \theta^*_{m(t)-1}(t) &=& \frac{\theta^*_{m(t)}(t) \omega^*_{m(t)}}{1 + \omega^*_{m(t)-1}} = \frac {\omega^*_{m(t)}}{(1+\omega^*_{m(t)-1})(1+\omega^*_{m(t)})} \\
\cdots \\
    \theta^*_2(t) &=& \frac{\theta^*_3(t) \omega^*_3}{1 + \omega^*_{2}} = \frac {\omega^*_3 \cdots \omega^*_{m(t)}}{(1+\omega^*_{2}) \cdots (1+\omega^*_{m(t)})} \\
    \theta^*_1(t) &=& \frac{\theta^*_2(t) \omega^*_2}{1 + \omega^*_{1}} = \frac {\omega^*_{2} \cdots \omega^*_{m(t)}}{(1+\omega^*_{1})\cdots(1+\omega^*_{m(t)})}. \\ 
\end{eqnarray*}
    We can now plug in the known values of $\omega^*_i$ and find that
    $$\bm{\theta}^{*}_{i}(t) = \left( \frac{z(i)}{z(m(t))} \right)^{\frac{1}{1 - p}} - \left( \frac{z(i -1)}{z(m(t))} \right)^{\frac{1}{1 - p}} \quad\forall 1 \leq i \leq m(t)$$
    This argument holds for any $m(t)\geq 1$, and hence we have fully specified the optimal allocation function.
\end{proof}

Taken together, the results of this section yield an expression for weighted flow time under the optimal allocation policy.  This expression is stated in Theorem \ref{thm:optt}.
\optt

\section{Discussion and Evaluation}
\label{sec:eval}
We now examine the impact of the results shown in Section \ref{sec:slow}.

Section \ref{sec:gen} shows how the results of Section \ref{sec:slow} can be applied to provide the optimal policy with respect to mean slowdown and mean flow time.

We perform a numerical evaluation of heSRPT in Section \ref{sec:numerical}.
We compare heSRPT to several competitor policies from the literature and show that heSRPT not only outperforms these policies as expected, but often reduces slowdown by over 30\%.

The above results apply in the common case where all jobs are present at time 0, but it is also interesting to consider the online case where jobs arrive over time.
To address the online case, we use heSRPT as the basis of a heuristic policy.
Section \ref{sec:numerical:online} compares our heuristic policy, Adaptive-heSRPT, to other heuristic allocation policies from the literature.
Adaptive-heSRPT vastly outperforms other heuristic policies in this online case.

\subsection{Applying heSRPT}
\label{sec:gen}
The optimality results of Section \ref{sec:opt} were stated in terms of any weighted flow time metric where weights favor small jobs.
Specifically, this class of metrics can be used to find the optimal policy with respect to several popular metrics including mean slowdown and mean flow time.
This is summarized in the following corollary.
\optflow
This corollary follows directly from the definitions of mean slowdown and mean flow time, respectively.
Specifically, our results hold in these cases because the weights used in both cases favor small jobs.

One might ask which, if any, of our results hold in the case where weights do not favor small jobs.
If weights do not favor small jobs it is easy to construct an example where one should not complete jobs in SJF order.
Giving a job of any size a sufficiently high weight can cause it to complete first under the optimal policy.
In addition to obviously contradicting Theorem \ref{thm:sjf}, the SJF completion order was exploited in the proof of Theorem \ref{thm:opt-omega}.
It is worth noting, however, that regardless of the weights used, the optimal policy will obey the scale free property of Theorem \ref{thm:sf}.

\begin{figure*}[t]
\centering
\includegraphics[width=.95\textwidth]{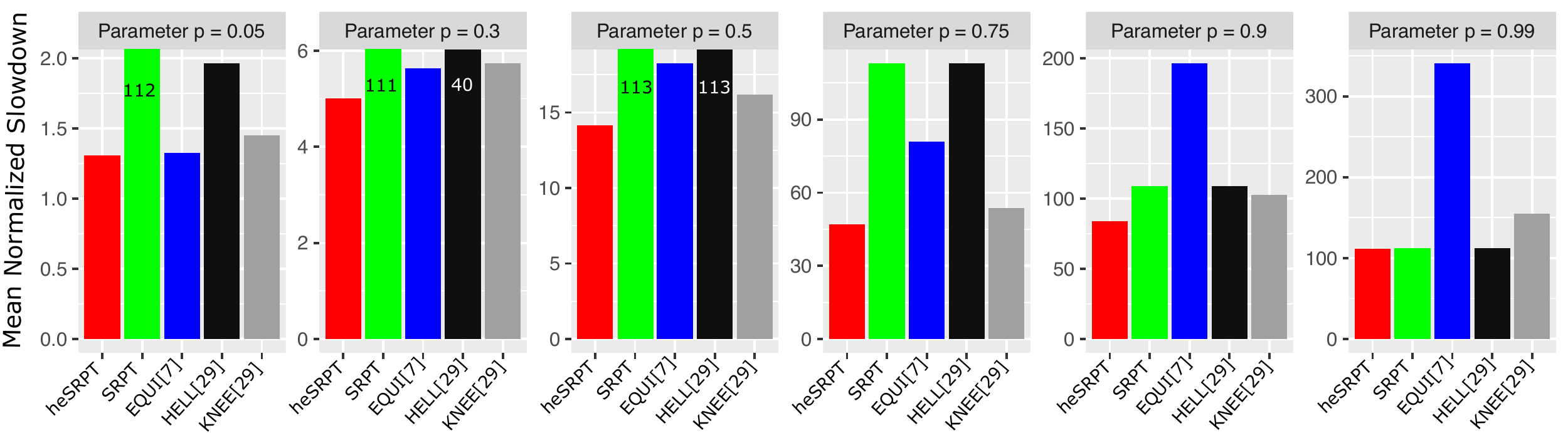}

    \caption{Mean slowdown of heSRPT and allocation policies from the literature in the offline setting with all jobs present at time 0.  Evaluation assumes $N=1,000,000$ servers and $M=500$ jobs whose sizes are Pareto($\alpha=.8$) distributed.  Each graph shows mean slowdown for one value of the speedup parameter, $p$, where $s(k)=k^p$.  heSRPT often dominates by over 30\%.}
\label{fig:bars}
\end{figure*}

\subsection{Numerical Evaluation: Offline Setting}
\label{sec:numerical}
We now compare the optimal mean slowdown under heSRPT with the mean slowdown under policies from the literature.

Our comparison continues to assume that all jobs are present at time 0.
While we have a closed form expression for the optimal mean slowdown under heSRPT, we wish to compare heSRPT to policies for which there is no closed form analysis.
Hence, we perform a numerical analysis of heSRPT and several policies from the literature.

We compare heSRPT to the following list of competitor policies:

     {\bf SRPT} allocates the entire system to the single job with shortest remaining processing time.  While SRPT is optimal when $p=1$, we expect SRPT to perform poorly when jobs make inefficient use of servers.  When all jobs are present at time 0, SRPT is equivalent to the RS policy \cite{wierman2005nearly} which minimizes mean slowdown in an M/G/1.

     {\bf EQUI} allocates an equal fraction of the servers to each job at every moment in time.  EQUI has been analyzed using competitive analysis \cite{edmonds2009scalably,Edmonds1999SchedulingIT} in similar models of parallelizable jobs, and was shown to be optimal in expectation when job sizes are unknown and exponentially distributed \cite{berg2018}.
        Other policies such as \emph{Intermediate-SRPT} \cite{im2016competitively} reduce to EQUI in our model where the number of jobs, $M$, is assumed to be less than the number of servers, $N$.

     {\bf HELL} is a heuristic policy proposed in \cite{lin2018model} which, similarly to heSRPT, tries to balance system efficiency with biasing towards short jobs.  HELL defines a job's efficiency to be the function $\frac{s(k)}{k}$.  HELL then iteratively allocates servers. In each iteration, HELL identifies the job which can achieve highest ratio of efficiency to remaining processing time, and allocates to this job the servers required to achieve this maximal ratio.  This process is repeated until all servers are allocated.  While HELL is consistent with the goal of heSRPT, the specific ratio that HELL uses is just a heuristic.

     {\bf KNEE} is the other heuristic policy proposed in \cite{lin2018model}.  
    KNEE defines a job's \emph{knee allocation} to be the number of servers for which the job's marginal reduction in run-time falls below some threshold, $\alpha$.  
    KNEE then iteratively allocates servers.
    In each iteration, KNEE identifies the job with the lowest knee allocation and gives this job its knee allocation.
    This process repeats until all servers are allocated.
    Because there is no principled way to choose this $\alpha$, we perform a brute-force search of the parameter space and present the results given the best $\alpha$ parameter we found.
    Hence, results for KNEE are an optimistic prediction of KNEE's performance.

%Given that this is the first paper to address the problem of minimizing mean slowdown for parallelizable jobs, these competitor policies were designed with the intent of minimizing mean flow time.
%Luckily, because our results also provide the optimal policy with respect to mean flow time (see Theorem \ref{thm:optflow}) we can also compare these competitor policies to the policy which minimizes mean flow time.
%We define the normalized slowdown of a job as 
%$$\mathcal{S}^P_i = \frac{T^P_i}{x_i/s(N)}.$$
%Instead of weighting a job's flow time by its running time on a single dedicated server, we weight each completion time by the job's running time on $N$ dedicated servers.
%This facilitates comparing slowdown across systems with different numbers of servers and across workloads with different levels of parallelizability.
%A policy minimizes mean normalized slowdown if and only if it minimizes mean slowdown.
%We will only consider mean slowdown for the majority of the paper, but we will examine mean normalized slowdown for the our numerical evaluation in Section \ref{sec:eval}.

We evaluate heSRPT and the competitor policies in a system of $N=1,000,000$ servers and $M=500$ jobs whose sizes are Pareto$(\alpha=.8)$ distributed.
The speedup parameter $p$ is set to values between 0.05 and 0.99.
Figure \ref{fig:bars} shows the results of this analysis.

heSRPT outperforms every competitor policy in every case as expected.
When $p$ is low, EQUI is within 1\% of heSRPT, but EQUI is over $3\times$ worse than heSRPT when $p=0.99$.
Conversely, when $p=0.99$, SRPT is nearly optimal.
However, SRPT is an order of magnitude worse than heSRPT when $p=0.05$.
While HELL performs similarly to SRPT in most cases, it is $50\%$ worse than optimal when $p=0.05$.
The KNEE policy is the best of the competitor policies that we consider.
In the worst case, when $p=0.99$, KNEE is roughly $50\%$ worse than heSRPT.
However, these results for KNEE required brute-force tuning of the allocation policy.
We examined several other job size distributions and in all cases we saw similar trends.

In \ref{sec:numericalflow} we compare these competitor policies to the optimal policy with respect to mean flow time.

\begin{figure*}[t]
\centering
\includegraphics[width=.95\textwidth]{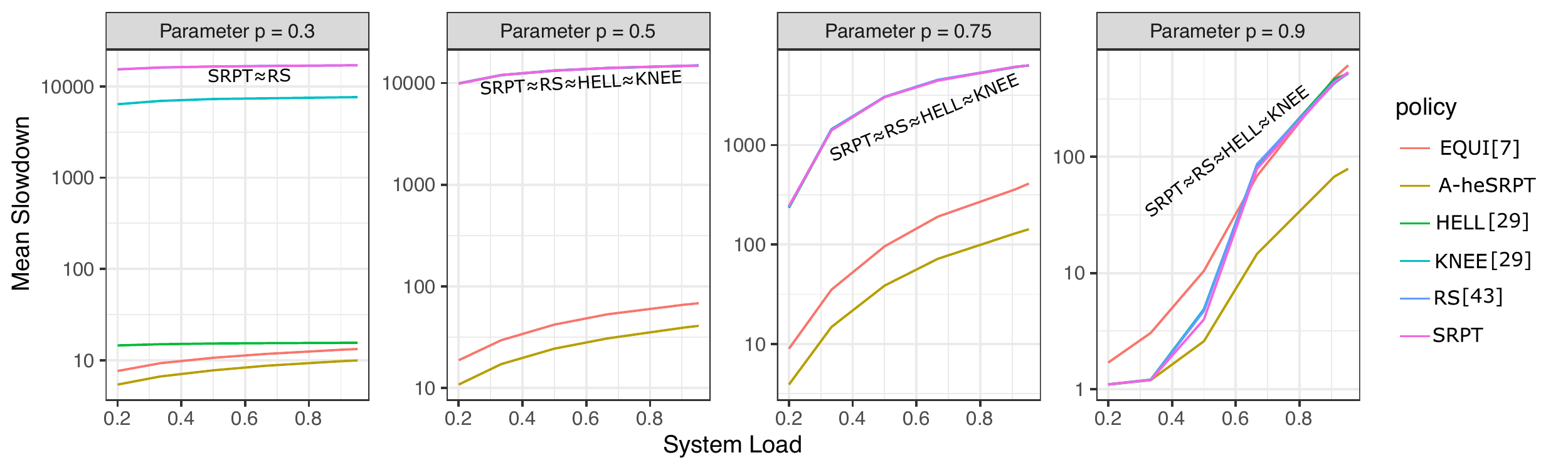}

    \caption{Mean slowdown of A-heSRPT and policies from the literature in the online setting. Simulations assume $N=10,000$ servers and a Poisson arrival process.  Job sizes are Pareto($\alpha=1.5$) distributed.  Each graph shows mean slowdown for one value of the speedup parameter, $p$, where $s(k)=k^p$.  A-heSRPT often dominates by an order of magnitude.}
\label{fig:online_bars}
\end{figure*}
\subsection{Numerical Evaluation: Online Setting}\label{sec:numerical:online}
While we have shown that heSRPT provides the optimal mean slowdown in the case where all jobs are present at time 0, minimizing the slowdown of a stream of arriving parallelizable jobs remains an open theoretical problem.

To understand the added complexity of the online setting, consider the problem of allocating servers to two jobs at time 0, while knowing that an arrival will occur at time 1.
The optimal policy might try to complete one job quickly and allow the arrival to compete for servers with the second job, or it might try and complete both of the original jobs before time 1 in order to reduce the flow time of the arriving job.
In general, it is no longer sufficient to compare two jobs and decide which to complete first.
One must also decide how many jobs to complete before the next arrival.
This prevents the results in this paper from generalizing cleanly to the online case.

We therefore use simulation to compare the performance of heSRPT to the competitor policies described above in an online setting where jobs arrive over time.
For this comparison, we define a heuristic online policy, inspired by heSRPT, which we refer to as Adaptive-heSRPT (A-heSRPT).
A-heSRPT recalculates the allocation to each job every time a job departs from \emph{or arrives to} the system, using the heSRPT allocations as defined in Theorem \ref{thm:opt}.
All of the other competitor policies we examine generalize to online setting in a similar way, reevaluating their allocation decisions on each arrival or departure.

For the sake of completeness, we also examine an additional competitor policy in the online case, the {\bf RS} policy \cite{wierman2005nearly} which is known to minimize mean slowdown online for non-parallelizable jobs in an M/G/1 system \cite{hyytia2012minimizing}.
The RS policy allocates all servers to the job with the lowest product of remaining size (R) times initial size (S).
RS is identical to SRPT if all jobs are present at time 0, but must be considered separately in the online case.

Figure \ref{fig:online_bars} shows that A-heSRPT outperforms all competitor policies in every case, \emph{often by several orders of magnitude}.
When jobs are highly parallelizable, policies such as SRPT and RS which aggressively favor small jobs can occasionally compete with A-heSRPT.
However, as load increases and the allocation decision becomes more complex, A-heSRPT vastly outperforms the competition.
Conversely, when jobs have a low level of parallelizability (low $p$), SRPT, RS, HELL and KNEE fail to maintain the efficiency of the system and perform poorly.
EQUI, which maximizes efficiency but does not favor small jobs, follows the opposite trend -- EQUI performs decently when $p$ is low, but it performs poorly when $p$ is high.
A-heSRPT is the only policy which is able to handle high and low values of $p$ over a range of system loads.

\section{Conclusion}
Modern data centers largely rely on users to decide how many servers to use to run their jobs.
When jobs are parallelizable, but follow a sublinear speedup function, allowing users to make allocation decisions can lead to a highly inefficient use of resources.
We propose to instead have the system control server allocations, and we derive the first optimal allocation policy which minimizes mean slowdown for a set of $M$ parallelizable jobs.
Our optimal allocation policy, heSRPT, leads to significant improvement over existing allocation policies suggested in the literature.
The key to heSRPT is that it finds the correct balance between overall system efficiency and favoring short jobs.
We derive an expression for the optimal allocation, at each moment in time, in closed form.

\section*{References}
\bibliography{bibshort}{}
\bibliographystyle{plain}
\clearpage
\appendix
%\section*{Appendix}
\section{Proof of Theorem \ref{thm:constant}}
\label{sec:constant}
\constant*
\begin{proof}
    Consider any time interval $[t_1, t_2]$ during which no job departs the system, and hence $m^*(t_1) = m^*(t_2)$.  
    Assume for contradiction that under the optimal policy $\bm{\theta}^*(t_1) \neq \bm{\theta}^*(t_2)$.  
    To produce a contradiction, we will show that the weighted flow time under the optimal policy can be improved by using a \emph{constant} allocation during the time interval $[t_1,t_2]$.
    The constant allocation we use is equal to the average value of $\bm{\theta}^*(t)$ during the interval $[t_1, t_2]$.

Specifically, consider the allocation function $\overline{\bm{\theta}^*(t)}$ where
$$\overline{\bm{\theta}^*(t)}=
    \begin{cases}
        \frac{1}{t_2 -t_1}\int_{t_1}^{t_2}\bm{\theta}^*(T)dT  &\forall t_1 \leq t \leq t_2\\
    \bm{\theta}^*(t)  &\mbox{otherwise}.
\end{cases}$$
Note that $\overline{\bm{\theta}^*(t)}$ is constant during the interval $[t_1,t_2]$.
    Furthermore, because $\sum_{i=1}^{m(t)}\theta_i^*(t) \leq 1$ for any time $t$, $\sum_{i=1}^{m(t)}\overline{\theta_i^*(t)} \leq 1$ at every time $t$ as well, and $\overline{\bm{\theta}^*(t)}$ is therefore a feasible allocation function.
Because the speedup function, $s$, is a strictly concave function, Jensen's inequality gives 
    $$\int_{t_1}^{t_2}s\left(\overline{\bm{\theta}^*(t)}\right)dt \geq \int_{t_1}^{t_2}s\left(\bm{\theta}^*(t)\right)dt$$
    and for at least one job, $i$, the above inequality will be strict.
Hence, the residual size of each job under the allocation function $\overline{\bm{\theta}^*(t)}$ at time $t_2$ is at most the residual size of that job under $\bm{\theta}^*(t)$ at time $t_2$, and the residual time of job $i$ is strictly less under the new policy at time $t_2$.
This guarantees that no jobs will have its flow time increased by this transformation and at least one job will have its flow time decreased by this transformation.
We have thus constructed a policy with a lower weighted flow time than the optimal policy, a contradiction.
\end{proof}

%\section{Proof of Theorem \ref{thm:helrpt}}
%\label{sec:moremake}
%\helrpt*
%\begin{proof}
%    We wish to find $\bm{\gamma}^*(t)$ such that, for any jobs $i$ and $j$,
%    $$\frac{x_i}{s(\gamma_i^*(t))} = \frac{x_j}{s(\gamma_j^*(t))} \qquad \forall t\leq T^*_{max}.$$
%    This implies that for any $i$ and $j$
%    $$\frac{x_i^{1/p}}{x_j^{1/p}} = \frac{\gamma^*_i(t)}{\gamma^*_j(t)}$$
%    Furthermore, we know that
%    $$\gamma_1^*(t) + \gamma_2^*(t) + \ldots + \gamma_M^*(t) = 1$$
%    and thus, for any job $i$,
%    $$\frac{x_i^{1/p} \gamma_i^*(t)}{x_i^{1/p}}+ \frac{x_2^{1/p} \gamma_i^*(t)}{x_i^{1/p}}+\ldots+ \frac{x_M^{1/p} \gamma_i^*(t)}{x_i^{1/p}}= 1$$
%    Rearranging, we have
%    $$\gamma^*_i(t) = \frac{x_i^{1/p}}{\sum_{j=1}^M x_j^{1/p}}$$
%    as desired.
%    Under this policy, the $T^*_{max}$ is equal to the completion time of any job, $i$, which is
%    $$\frac{x_i}{s\biggl(\frac{x_i^{1/p}}{\sum_{j=1}^{M}x_j^{1/p}}\biggr)}= ||\bm{X}||_{1/p}$$
%    as desired.
%
%\end{proof}
\section{Proof of Lemma \ref{lem:scaling}}
\label{sec:morescaling}
\scaling*
\begin{proof}
    Consider the policy $\bm{\theta}'(t)$ in a system which runs at $(1-\beta)^p = s(1-\beta)$ times the speed of the original system on $[0, T]$.
    The service rate of any job in this system on $[0,T]$ is given by
    $$ s(1-\beta) \cdot s(\bm{\theta}'(t) \cdot N) = s((1-\beta) \cdot \bm{\theta}'(t) \cdot N) = s(\bm{\theta}(t) \cdot N).$$
    Since, at any time  $t'\in[0,T]$, the service rate for any job $i$ is the same under $\bm{\theta}(t)$ as it is under $\bm{\theta}'(t)$ in a system which is $(1-\beta)^p$ times as fast, the same amount of work is done on job $i$ by time $T$ in both systems.
\end{proof}
\section{Proof of Lemma \ref{lem:increase}}
\label{sec:addl}
%\begin{restatable}{lem}{om}
%    \label{lem:om}
%We define
%    $$\omega_{k}^P = \frac{1}{\left(\frac{z(k)}{z(k-1)}\right)^{\frac{1}{1-p}} - 1} \qquad \forall 1 < k \leq M$$
%    and $\omega^P_1 = 0$. Then $\omega^P_k$ is increasing in $k$ for $1 \leq k \leq M$. 
%\end{restatable}
%\begin{proof}
%    Note that since $\frac{k}{k-1}$ is decreasing in $k$,
%    the denominator of the expression for $\omega^P_k$ is decreasing for $1 < k \leq M$.
%    Hence $\omega^P_k$ is increasing when $1 < k \leq M$, and since $\omega^P_1 = 0 < \omega^P_2$, $\omega^P_k$ is increasing for $1 \leq k \leq M$.
%\end{proof}

\begin{restatable}{lem}{increasing}\label{lem:increase}
    Let $\bm{\Theta}^P(t)$ be the allocation function which satisfies the following first-order conditions:
    $$\frac{\partial \Sl{P}}{\partial \omega^P_{k}} = 0 \qquad \forall 1 \leq k \leq M.$$
    Then for any $t$, $\Theta^P_k(t)$ is increasing in $k$ for $1 \leq k \leq m(t)$.
\end{restatable}
\begin{proof}
    Following the same argument as Theorem \ref{thm:opt}, we can see that the expression for $\bm{\Theta}^P(t)$ is
    $$\Theta^P_i(t) = \left(\frac{z(i)}{z(m(t))}\right)^{\frac{1}{1-p}} - \left(\frac{z(i-1)}{z(m(t))}\right)^{\frac{1}{1-p}} \qquad \forall 1 \leq i \leq m(t).$$
    Note that $\frac{1}{1-p} > 1$, so $i^{\frac{1}{1-p}}$ is convex in $i$.
    We know that 
    $$z(i) - z(i-1) = w_i$$
    and 
    $$z(i+1) - z(i) = w_{i+1}.$$
    Furthermore, $w_{i+1} \geq w_i$.
    Hence, by convexity,
    \begin{align*}
        z(i)^{\frac{1}{1-p}} - z(i-1)^{\frac{1}{1-p}} &< \left(z(i) + w_i\right)^{\frac{1}{1-p}} - \left(z(i-1) + w_i\right)^{\frac{1}{1-p}}\\
        & < \left(z(i) + w_{i+1}\right)^{\frac{1}{1-p}} - z(i)^{\frac{1}{1-p}}\\
        &< z(i+1)^{\frac{1}{1-p}} - z(i)^{\frac{1}{1-p}}
    \end{align*}
    This implies that $\Theta^P_i(t)$ is increasing in $i$ for $1 \leq i \leq m(t)$.
\end{proof}

%\begin{restatable}{lem}{coef}
%We define
%    $$\omega_{k}^P = \frac{1}{\left(\frac{k}{k-1}\right)^{\frac{1}{1-p}} - 1} \qquad \forall 1 < k \leq M.$$
%    and $\omega_1^P = 0$.  Then the expression
%    $$\Delta(k) = k  s(1 + \omega^P_{k})  - (k-1) s(\omega^P_{k})$$
%    is increasing in $k$.
%\end{restatable}
%\begin{proof}
%    Let $c=\frac{1}{1-p}$.  We can rewrite $\Delta(k)$ as
%    $$(k^c - (k-1)^c)^{1-p}.$$
%    We then have that
%    $$\frac{d}{dk} \Delta(k) = (1-p)(k^c - (k-1)^c)^{-p}(c(k^{c-1}-(k-1)^{c-1})).$$
%    This derivative is positive for all $k$, and hence $\Delta(k)$ is increasing in $k$.
%\end{proof}
\begin{figure*}[ht!]
\centering
\includegraphics[width=.95\textwidth]{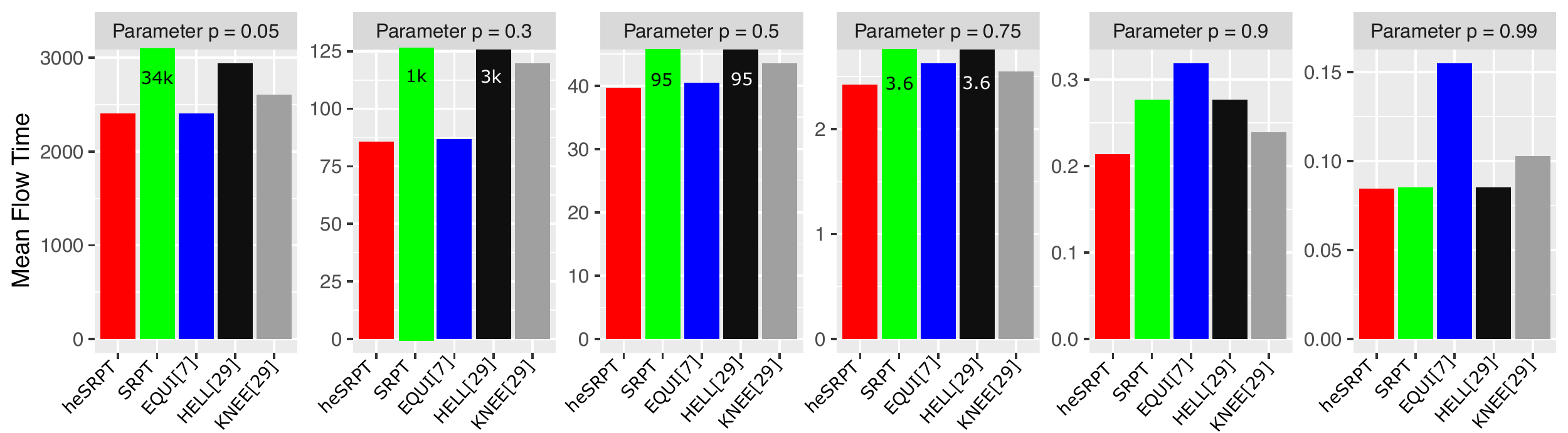}

    \caption{A comparison of the optimal mean flow time under heSRPT to other allocation policies found in the literature.  Each policy is evaluated on a system of $N=1,000,000$ servers and a set of $M=500$ jobs whose sizes are drawn from a Pareto distribution with shape parameter $.8$.  Each graph shows the mean flow time under each policy with various values of the speedup parameter, $p$, where the speedup function is given by $s(k)=k^p$.  heSRPT outperforms every competitor by at least 30\% in at least one case.}
\label{fig:flowbars}
\end{figure*}

\section{Numerical Evaluation for Mean Flow Time}
\label{sec:numericalflow}

\label{sec:compare}
Section \ref{sec:gen} describes how to use heSRPT to obtain the optimal policy with respect to mean flow time.
Because the competitor policies in Section \ref{sec:numerical} were designed to minimize mean flow time, we now compare the optimal mean flow time under heSRPT to the mean flow time of these competitor policies.
The competitor policies remain unchanged from their descriptions in Section \ref{sec:eval}, and the conditions of our analysis (numbers of servers, job size distribution, speedup function) remain the same as is Section \ref{sec:eval} as well.

Figure \ref{fig:flowbars} shows the results of our analysis.
We see that the results with respect to mean flow time are generally similar to the results of Section \ref{sec:eval}.
heSRPT once again outperforms every competitor policy by at least $30\%$ in at least one case.
Hence, the results shown in Section \ref{sec:eval} are not only caused by the fact that the competitor policies optimize for a different metric.
Even when comparing policies with respect to mean flow time, each of the competitor policies can be far from optimal.

\end{document}